\documentclass[a4paper,UKenglish,cleveref, autoref, numberwithinsect]{lipics-v2021}

\title{On the Power of Graphical Reconfigurable Circuits} 

\author{Yuval Emek}{Technion - Israel Institute of Technology, Israel }{yemek@technion.ac.il}{https://orcid.org/0000-0002-3123-3451}{}

\author{Yuval Gil}{Technion - Israel Institute of Technology, Israel }{yuval.gil@campus.technion.ac.il}{https://orcid.org/0009-0007-7762-3029}{}

\author{Noga Harlev}{Technion - Israel Institute of Technology, Israel }{snogazur@campus.technion.ac.il}{}{}

\authorrunning{Y. Emek, Y. Gil, N. Harlev} 

\Copyright{Yuval Emek, Yuval Gil, and Noga Harlev} 


\nolinenumbers 

\EventEditors{Dan Alistarh}
\EventNoEds{1}
\EventLongTitle{38th International Symposium on Distributed Computing (DISC 2024)}
\EventShortTitle{DISC 2024}
\EventAcronym{DISC}
\EventYear{2024}
\EventDate{October 28--November 1, 2024}
\EventLocation{Madrid, Spain}
\EventLogo{}
\SeriesVolume{319}
\ArticleNo{23}

\usepackage{booktabs}
\usepackage{algorithm,algpseudocode}
\algtext*{EndWhile} 
\algtext*{EndIf} 
\algtext*{EndFor} 

\newcommand{\Sect}{Sec.}
\newcommand{\Thm}{Thm.}
\newcommand{\Lem}{Lem.}
\newcommand{\Obs}{Obs.}
\newcommand{\Cor}{Cor.}

\newcommand{\Proc}{Proc.}

\newcommand{\Alg}{\ensuremath{\mathtt{Alg}}}
\newcommand{\Reals}{\mathbb{R}}

\newcommand{\Degree}{\operatorname{deg}}
\newcommand{\Integers}{\mathbb{Z}}
\newcommand{\Pins}{\mathit{P}}
\newcommand{\PinPart}{\mathcal{R}}
\newcommand{\RefTrnsCl}{\operatorname{tc}}
\newcommand{\Circuits}{\mathcal{C}}
\newcommand{\Distance}{d_{G}}
\providecommand{\Pr}{}
\renewcommand{\Pr}{\mathbb{P}}

\begin{document}
	
	\maketitle
	
	\begin{abstract}
		We introduce the \emph{graphical reconfigurable circuits (GRC)} model as an
		abstraction for distributed graph algorithms whose communication scheme is
		based on local mechanisms that collectively construct long-range
		reconfigurable channels (this is an extension to general graphs of a
		distributed computational model recently introduced by Feldmann et al.\ (JCB
		2022) for hexagonal grids).
		The crux of the GRC model lies in its modest assumptions:
		(1)
		the individual nodes are computationally weak, with state space bounded
		independently of any global graph parameter;
		and
		(2)
		the reconfigurable communication channels are highly restrictive, only
		carrying information-less signals (a.k.a.\ \emph{beeps}).
		Despite these modest assumptions, we prove that GRC algorithms can solve many
		important distributed tasks efficiently, i.e., in polylogarithmic time.
		On the negative side, we establish various runtime lower bounds, proving that
		for other tasks, GRC algorithms (if they exist) are doomed to be
		slow.
	\end{abstract}
	\section{Introduction}
	\label{sec:introduction}
	The \emph{reconfigurable circuits} model was introduced recently by Feldmann
	et al.~\cite{feldmannPSD2022} and studied further by Padalkin et al.\
	\cite{padalkinSW2022,
		PadalkinS2024shortest-path}.
	It extends the popular \emph{geometric amoebot} model for (synchronous)
	distributed algorithms running in the hexagonal grid by providing them with an
	opportunity to form long-range communication channels.
	This is done by means of a distributed mechanism that allows each node to bind
	together a subset of its incident edges (which can be thought of as installing
	internal ``wires'' between the corresponding ports);
	the long-range channels, a.k.a.\ \emph{circuits}, are then formed by taking the
	transitive closure of these local bindings (see \Sect{}~\ref{sec:model} for
	details).
	The circuits serve as \emph{beeping channels}, enabling their participating
	nodes to communicate via information-less signals.
	The crux of the model is that the distributed mechanism that controls the
	circuit formation is invoked in every round (of the synchronous execution) so
	that the circuits can be reconfigured.
	
	In contrast to the original geometric amoebot model which is tailored
	specifically to planarly embedded (hexagonal) grids, the reconfigurable
	circuits model can be naturally generalized to arbitrary graph topologies.
	The starting point of the current paper is the formulation of such a
	generalization that we refer to as the \emph{graphical reconfigurable circuits
		(GRC)} model (formally defined in \Sect{}~\ref{sec:model}).
	
	An important feature of the GRC model is that it is \emph{uniform}:
	the actions of each node $v$ in the (general) communication graph $G$ are
	dictated by a (possibly randomized) state machine whose description is fully
	determined by the degree of $v$ (and the local input provided to $v$ if there
	is such an input), independently of any global parameter of $G$
	\cite{Angluin1980local}.
	A clear advantage of uniform algorithms is that they can be deployed in a
	``one size fits all'' fashion, without any global knowledge of the graph on
	which they run.
	We further require that the aforementioned state machines admit a finite
	description, which means, in particular, that the state space of the state
	machines are bounded independently of any global graph parameter.
	This requirement is an obvious necessary condition for practical
	implementations;
	we subsequently refer to uniform distributed algorithms subject to this
	requirement as \emph{boundedly uniform}.
	
	Combining the bounded uniformity with the light demands of the beeping
	communication scheme, demands which are known to be easy to meet in practice
	\cite{CornejoK2010deploying,
		FluryW2010slotted},
	we conclude that the GRC model provides an abstraction for distributed
	(arbitrary topology) graph algorithms that can be implemented over devices
	with slim computation and communication capabilities.
	In particular, the GRC model may open the gate for a rigorous investigation of
	distributed algorithms operating in (natural or artificial) biological
	cellular networks whose communication mechanism is based on bioelectric
	signaling, known to be the basis for long range (low latency) communication in
	such networks.
	
	The main technical contribution of this paper is the design of GRC algorithms
	for various classic distributed tasks that terminate in polylogarithmic time.
	Some of these tasks (e.g., the construction of a minimum spanning tree) are
	inherently global and are known to be subject to congestion bottlenecks,
	thus demonstrating that despite their limited computation and communication
	power, GRC algorithms can overcome both ``locality'' and ``bandwidth''
	barriers.
	In fact, as far as we know, these are the first distributed algorithms that
	solve such tasks in polylogarithmic time under any boundedly uniform model.
	
	While GRC algorithms can bypass the congestion bottlenecks of some distributed
	tasks, other tasks turn out to be much harder:
	We prove that under certain conditions, runtime lower bound constructions,
	developed originally for the CONGEST model \cite{peleg2000}, can be
	translated, almost directly, to the GRC model, thus establishing runtime lower
	bounds for a wide class of tasks.

	\subsection{The GRC Model}
	\label{sec:model}
	In the current section, we introduce the distributed computational model used
	throughout this paper, referred to as the \emph{graphical reconfigurable
		circuits (GRC)} model.
	A GRC algorithm \Alg{} runs over a (finite simple) undirected graph
	$G = (V, E)$
	so that each node
	$v \in V$
	is associated with its own copy of a (possibly randomized) state machine
	defined by \Alg{};
	for clarity of the exposition, we often address node $v$ and the state machine
	that dictates $v$'s actions as the same entity (our intention will be clear
	from the context).
	
	We adopt the \emph{port numbering} convention
	\cite{Angluin1980local,
		HellaJKLLLSV2015weak-models}
	stating that from the perspective of a node
	$v \in V$,
	each edge
	$e \in E(v)$
	is identified by a unique port number taken from the set
	$\{ 1, \dots, \Degree(v) \}$.\footnote{%
		Given an edge subset
		$F \subseteq E$
		and a node
		$v \in V$,
		we denote the set of edges in $F$ incident on $v$ by
		$F(v) = \{ e \in F \mid e \ni v \}$
		and the degree of $v$ by
		$\Degree(v) = | E(v) |$.
	}
	Every edge
	$e \in E$
	is associated with $k$ \emph{pins}, where
	$k \in \Integers_{> 0}$
	is a constant determined by the algorithm designer;\footnote{%
		For the (asymptotic) upper bounds established in the current paper, it is
		actually sufficient to use
		$k = 1$
		pins per edge.
		However, this is not true in general (see, e.g.,
		\cite[Sections 3.4 and 4.4]{feldmannPSD2022})
		and regardless, using multiple (yet,
		$O (1)$)
		pins per edge often facilitates the algorithm's exposition.
		In any case, we do not make an effort to optimize the value of $k$.}
	these pins are represented as pairs of the form
	$p = (e, i)$
	for
	$i \in [k]$.
	Let
	$\Pins = E \times [k]$
	denote the set of all pins.
	For a node
	$v \in V$,
	let
	$\Pins(v) = E(v) \times [k]$
	denote the set of pins associated with the edges incident on $v$.
	The GRC model is defined so that for each pin
	$p = (e, i) \in \Pins(v)$,
	node $v$ is aware of the (local) port number of edge $e$ as well as the
	(global) index
	$i \in [k]$.
	In particular, the other endpoint of edge $e$ agrees with $v$ on the index $i$
	of pin $p$ although the two nodes may identify $e$ by different port numbers.
	
	The execution of algorithm \Alg{} advances in synchronous \emph{rounds}.
	Each round
	$t = 0, 1, \dots$
	is associated with a partition $\Circuits^{t}$ of the pin set $\Pins$ into
	non-empty pairwise disjoint parts, called \emph{circuits}.
	The partition $\Circuits^{0}$ is defined so that each pin forms its own
	singleton circuit;
	for
	$t \geq 1$,
	the partition $\Circuits^{t}$ is determined by the nodes according to a
	distributed mechanism explained soon.
	
	For a round
	$t \geq 0$,
	a node
	$v \in V$
	is said to \emph{partake} in a circuit
	$C \in \Circuits^{t}$
	if
	$\Pins(v) \cap C \neq \emptyset$.
	Let
	$\Circuits^{t}(v) = \{ C \in \Circuits^{t} \mid \Pins(v) \cap C \neq \emptyset
	\}$
	denote the set of circuits in which node $v$ partakes.
	
	The communication scheme of the GRC model is defined on top of the circuits so
	that each circuit
	$C \in \Circuits^{t}$
	serves (during round $t$) as a \emph{beeping channel}
	\cite{CornejoK2010deploying}
	for the nodes that partake in $C$.
	Before getting into the specifics of this communication scheme, let us explain
	how the partition
	$\Circuits^{t}$
	is formed based on the actions of the nodes in round
	$t - 1$.
	
	Fix some round
	$t \geq 1$.
	Towards the end of round
	$t - 1$,
	each node
	$v \in V$
	decides on a partition $\PinPart^{t}(v)$ of $\Pins(v)$, referred to as the
	\emph{local pin partition} of $v$.
	Let $\mathcal{L}^{t}$ be the symmetric binary relation over $\Pins$ defined so
	that pins
	$p = (e, i)$
	and
	$p' = (e', i')$
	are related under $\mathcal{L}^{t}$ (i.e.,
	$(p, p'), (p', p) \in \mathcal{L}^{t}$)
	if and only if there exists a node
	$v \in V$
	(incident on both $e$ and $e'$) such that $p$ and $p'$ belong to the same
	part of $\PinPart^{t}(v)$.
	Let $\RefTrnsCl(\mathcal{L}^{t})$ be the reflexive transitive closure of
	$\mathcal{L}^{t}$, which is, by definition, an equivalence relation over
	$\Pins$.
	The circuits in $\Circuits^{t}$ are taken to be the equivalence classes of
	$\RefTrnsCl(\mathcal{L}^{t})$.
	See Figure~\ref{figure:grc} for an illustration.\footnote{%
		As presented by Feldmann et al.~\cite{feldmannPSD2022}, the physical
		interpretation of the abstract circuit forming process is that each node $v$
		internally ``wires'' all pins belonging to the same part
		$R \in \PinPart^{t}(v)$
		to each other, thus ensuring that a signal transmitted over one pin in $R$ is
		disseminated to all pins in $R$ (and through them, to the entire circuit that
		contains $R$).}
	
		\begin{figure}
		\centering
		\includegraphics[width=0.8\textwidth]{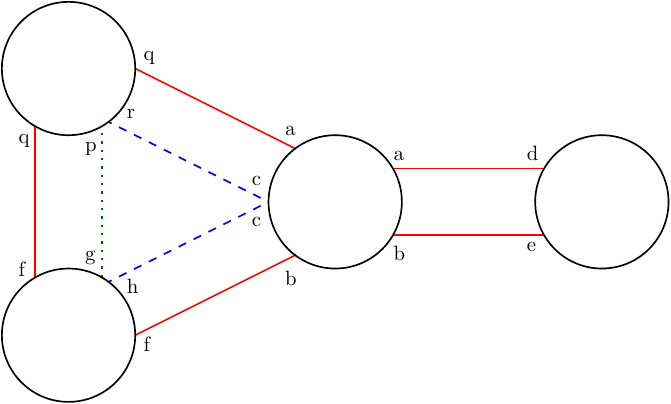}
		\caption{\label{figure:grc}%
			The circuits formed on a communication graph by the local node decisions.
			The graph includes $4$ nodes, depicted by the black cycles, and $4$ edges (not
			shown explicitly in the figure), each one of them is associated with
			$k = 2$
			pins, depicted by the straight lines.
			The local pin partitions are presented by the lower-case letters.
			These local pin partitions result in forming three circuits, consisting of
			the red (solid) pins,
			the blue (dashed) pins,
			and
			the green (dotted) pin.
		}
	\end{figure}
	
	We are now ready to formally define the operation of each node
	$v \in V$
	in round
	$t = 0, 1, \dots$
	This includes the following three steps, where we denote the state of $v$ in
	round $t$ by $S^{t}(v)$:
	\\
	(1)
	Node $v$ decides (possibly in a probabilistic fashion), based on $S^{t}(v)$,
	on a pin subset
	$B^{t}(v) \subseteq \Pins(v)$
	and \emph{beeps} --- namely, emits an information-less signal --- on every pin
	in $B^{t}(v)$;
	we say that $v$ \emph{beeps} on a circuit
	$C \in \Circuits^{t}(v)$
	if $v$ beeps on (at least) one of the pins in $C$.
	\\
	(2)
	For each pin
	$p \in \Pins(v)$,
	node $v$ obtains a bit of information revealing whether at least one node
	beeps (in the current round) on the (unique) circuit
	$C \in \Circuits^{t}$
	to which $p$ belongs.
	\\
	(3)
	Node $v$ decides (possibly in a probabilistic fashion), based on $S^{t}(v)$
	and the information obtained in step (2), on the next state
	$S^{t + 1}(v)$
	and the next local pin partition
	$\PinPart^{t + 1}(v)$.
	\\
	We emphasize that for each circuit
	$C \in \Circuits^{t}(v)$
	and pin
	$p \in \Pins(v) \cap C$,
	node $v$ can distinguish, based on the information obtained in step (2) for
	$p$, between the scenario in which zero nodes beep on $C$ and the scenario in
	which a positive number of nodes beep on $C$, however, node $v$ cannot tell
	how large this positive number is.
	In fact, if $v$ itself decides (in step (1)) to beep on pin $p$, then $v$ does
	not obtain any meaningful information from $p$ in step (2) (in the
	beeping model terminology \cite{AfekABCHK2013beeping}, this is referred to as
	lacking ``sender collision detection'').\footnote{%
		\label{footnote:far-away-local-pin-partitions}%
		The reader may wonder why the decisions made in step (1) and the information
		obtained in step (2) are centered on the pins in $\Pins(v)$, rather than on
		the circuits in $\Circuits^{t}(v)$.
		The reason is that node $v$ is not necessarily aware of the partition induced
		on $\Pins(v)$ by $\Circuits^{t}(v)$ (i.e., the exact assignment of the pins in
		$\Pins(v)$ to the circuits in $\Circuits^{t}(v)$);
		indeed, the latter partition depends on the local pin partitions
		$\PinPart^{t}(u)$ of other nodes
		$u \in V$,
		some of which may be far away from $v$.
		For example, in Figure~\ref{figure:grc}, the local pin partition of the
		rightmost node separates between its two incident pins;
		nevertheless, both pins belong to the same (red) circuit due to local pin
		partitions decided upon in the other side of the graph.}
	
	An important feature of the GRC model is that \Alg{} is required to be
	\emph{boundedly uniform}, namely, the number of states in the state machine
	associated with a node
	$v \in V$,
	as well as the description of the transition functions that determine the next
	state
	$S^{t + 1}(v)$
	and the next local pin partition
	$\PinPart^{t + 1}(v)$,
	are finite and fully determined by the local parameters of $v$, independently
	of any global parameter of the graph $G$ on which \Alg{} runs.
	These local parameters include the degree $\Degree(v)$ of $v$ and, depending
	on the specific task, any local input provided to $v$ at the beginning of the
	execution (e.g., the weights of the edges incident on $v$).\footnote{%
		To maintain strict uniformity, we adhere to the convention that numerical
		values included in the local inputs (e.g., edge weights) are encoded as
		bitstrings without ``leading zeros'', thus ensuring that the length of such a
		bitstring by itself does not reveal any global information.}
	In particular, node $v$ does not ``know'' (and generally, cannot encode)
	the number
	$n = |V|$
	of nodes,
	the number
	$m = |E|$
	of edges,
	the maximum degree
	$\Delta = \max_{v \in V} \Degree(v)$,
	or
	the diameter
	$D = \max_{u, v \in V} \Distance(u, v)$.\footnote{%
		The notation
		$\Distance(u, v)$
		denotes the distance (in hops) between nodes $u$ and $v$ in $G$.}
	Notice that the uniformity in $n$ means that the nodes are also
	\emph{anonymous}, i.e., they do not (and cannot) have unique identifiers.
	
	The primary performance measure applied to our algorithms is their
	\emph{runtime} defined to be the number of rounds until termination.
	When the algorithm is randomized, its runtime may be a random variable, in
	which case we aim towards bounding it whp.\footnote{%
		An event $A$ holds \emph{with high probability (whp)} if
		$\Pr (A) \geq 1 - n^{-c}$
		for an arbitrarily large constant $c$.}
	
	\subparagraph{Relation to CONGEST.}
	An adversity faced by GRC algorithms is the limited amount of information that
	can be sent/received by each node in a single round.
	Such limitations lie at the heart of the popular \emph{CONGEST}
	\cite{peleg2000} model that operates in synchronous message passing rounds,
	using messages of size $B$, where the typical choices for $B$ are
	$B = O (1)$,
	$B = \Theta (\log n)$,
	or
	$B = \operatorname{polylog} (n)$
	(by definition, the uniform version of CONGEST adopts the former choice).
	An important point of similarity between the two models is that per round,
	both CONGEST and GRC algorithms can communicate
	$\tilde{O} (s)$
	bits of information over a cut of size $s$.\footnote{%
		The asymptotic notations
		$\tilde{O}(\cdot)$
		and
		$\tilde{\Omega}(\cdot)$
		hide
		$\operatorname{polylog}(n)$
		expressions.}
	As explained in \Sect{}~\ref{sec:preliminaries}, from the perspective of
	message exchange per se (regardless of local computation), $T$ CONGEST rounds
	can be simulated by
	$O (\log n + T \cdot B)$
	GRC rounds whp, so, ignoring the additive logarithmic term, GRC algorithms are
	at least as strong as the boundedly uniform version of CONGEST algorithms.
	In fact, they are strictly stronger:
	the crux of GRC algorithms is that they enjoy the advantage of reconfigurable
	long-range communication channels (though highly restrictive ones);
	this advantage materializes in some of the GRC algorithms developed in the
	sequel whose runtime is significantly smaller than their corresponding (not
	necessarily uniform) CONGEST lower bounds.
	
	\subsection{Our Contribution}
	\label{sec:contribution}
	The main takeaway from this paper is that many important distributed tasks
	admit highly efficient GRC algorithms --- see Table~\ref{table:upper-bounds}.
	Notice that with the exception of the sparse spanner construction, all tasks
	mentioned in Table~\ref{table:upper-bounds} admit
	$\tilde{\Omega} (\sqrt{n} + D)$
	runtime lower bounds under the (not necessarily uniform) CONGEST model
	\cite{dasHKKNPPW2011,
		PelegR00},
	demonstrating that reconfigurable beeping channels are a powerful tool even
	for boundedly uniform algorithms.
	
	\begin{table*}
		\begin{tabular}{l|l|c}
			\toprule
			\multicolumn{2}{c|}{task} &
			runtime
			\\
			\midrule
			\multirow{2}{*}{construction} &
			minimum spanning tree
			(integral edge weights
			$\in [1, W]$) &
			$O (\log (n) \cdot \log (n + W))$
			\\
			\cmidrule{2-3}
			&
			$(2 \kappa - 1)$-spanner
			with
			$O (n^{1 + (1 + \varepsilon) / \kappa})$
			edges in expectation &
			$O (\kappa + \log n)$
			\\
			\midrule
			\multirow{2}{*}{verification} &
			minimum spanning tree
			(integral edge weights
			$\in [1, W]$) &
			$O (\log (n) \cdot \log (n + W))$
			\\
			\cmidrule{2-3}
			&
			\begin{minipage}{0.6\textwidth}
				simple path,
				connectivity,
				$(s, t)$-connectivity,
				connected spanning subgraph,
				cut,
				$(s, t)$-cut,
				Hamiltonian cycle,
				$e$-cycle containment,
				edge on all
				$(s, t)$-paths
			\end{minipage}
			&
			$O (\log n)$
			\\
			\bottomrule
		\end{tabular}
		\caption{\label{table:upper-bounds}%
			Our runtime upper bounds.
			The corresponding GRC algorithms are randomized and their correctness and
			runtime guarantees hold whp;
			the one exception is the spanner construction, where the number of edges is
			bounded in expectation.
		}
	\end{table*}
	
	The polylogarithmic runtime upper bounds presented in
	Table~\ref{table:upper-bounds} imply that the
	$\tilde{\Omega} (\sqrt{n} + D)$
	CONGEST lower bounds for the corresponding tasks fail to transfer to
	the GRC model (refer to \Sect{}~\ref{sec:inapplicable-reductions} for further
	discussion of this ``failed transfer'').
	CONGEST lower bounds for other distributed tasks on the other hand do
	transfer, almost directly, to GRC.
	Indeed, we develop a generic translation, from CONGEST runtime lower bounds to
	GRC runtime lower bounds, which applies to a large class of CONGEST lower bound
	constructions --- see Table~\ref{table:lower-bounds} (in
	\Sect{}~\ref{sec:lower-bounds}) for a sample of the results obtained through
	this translation.
	
	\subsection{Paper's Outline}
	The remainder of this paper is organized as follows.
	We start in \Sect{}~\ref{sec:overview} with a discussion of the main technical
	challenges encountered towards establishing our results and the ideas used to
	overcome them.
	\Sect{}~\ref{sec:preliminaries} introduces some preliminary definitions, as
	well as several basic procedures used in the later technical sections.
	The GRC algorithms promised in Table~\ref{table:upper-bounds} for the tasks of
	constructing a minimum spanning tree and a spanner are presented and analyzed
	in \Sect{} \ref{sec:mst} and \ref{sec:spanner}, respectively.
	\Sect{}~\ref{sec:verification-tasks} is dedicated to the algorithms for the
	verification tasks promised in the bottom half of that table.
	Our GRC runtime lower bounds (as discussed in \Sect{}~\ref{sec:contribution})
	are established in \Sect{}~\ref{sec:lower-bounds}.
	We conclude in \Sect{}~\ref{sec:related-work} with a discussion of additional
	related work.
	(Throughout, missing proofs are deferred to
	Appendix~\ref{appendix:missing-proofs}.)
	
	\section{Technical Overview}
	\label{sec:overview}
	In this section, we discuss the different challenges that arise in our upper
	and lower bound constructions and present a brief overview of the
	technical ideas used to overcome these challenges;
	see \Sect{}\ \ref{sec:mst}, \ref{sec:spanner}, and \ref{sec:lower-bounds} for
	the full details.
	(The techniques employed in \Sect{}~\ref{sec:verification-tasks} for the
	verification tasks are similar to those developed in \Sect{}~\ref{sec:mst} for
	the MST algorithm.)
	
	\subparagraph{Minimum Spanning Tree.} 
	The minimum spanning tree (MST) construction follows the structure of
	Boruvka's classic algorithm \cite{boruvka1926}.
	The algorithm maintains a partition of the node set into \emph{clusters} that
	correspond to the connected components of the subgraph induced by the edges
	which were already selected for the MST.
	It operates in phases, where the main algorithmic task in a phase is to
	identify a \emph{lightest outgoing edge} for each cluster.
	The clusters are then merged over the identified edges, adding those edges to
	the output edge set.
	
	If the edge weights are distinct, then no cycles are formed by the cluster
	merging process and Boruvka's algorithm is guaranteed to return an MST of the
	original graph.
	This well known fact is utilized by the existing distributed implementations
	of Boruvka's algorithm that typically use the unique node IDs to ``enforce''
	distinct edge weights.
	
	Unfortunately, obtaining distinct edge weights under our boundedly uniform
	model is hopeless.
	This means that the set $L$ of lightest outgoing edges (of all clusters)
	cannot be safely added to the output edge set without the risk of forming
	cycles, thus forcing us to come up with an alternative mechanism.
	The key technical idea here is a procedure that runs in each phase
	independently and constructs (whp) a \emph{total order} $\mathcal{T}$ over
	the set $L$.
	Following that, we identify a $\mathcal{T}$-minimal outgoing edge for each
	cluster and perform the cluster merger over the identified edges.
	As we prove in \Sect{}~\ref{sec:mst}, selecting the $\mathcal{T}$-minimal
	outgoing edges ensures that no cycles are formed, resulting in a valid MST.
	Notice that for this argument to work, it is crucial that $\mathcal{T}$ is
	defined \emph{globally} over all edges in $L$ which is ensured by a careful
	design of the aforementioned procedure.
	
	\subparagraph{Spanner.} 
	The spanner construction is based on the elegant random shifts method of
	\cite{MillerPX13}.
	Particularly, the idea is similar to the distributed algorithm of
	\cite{ForsterGV21} that uses random shifts to obtain a
	$(2 \kappa - 1)$-spanner
	of expected size
	$O(n^{1 + 1 / \kappa})$.
	The heart of the random shift method is a probabilistic clustering process
	based on a random variable $\delta_{v}$ drawn independently by each node
	$v \in V$.
	Specifically, in \cite{ForsterGV21}, each node
	$v \in V$
	samples $\delta_{v}$ from the capped geometric distribution (see
	\Sect~\ref{sec:preliminaries} for a definition) with
	parameters
	$p = 1 - n^{-1 / \kappa}$
	and
	$r = \kappa - 1$.
	The main challenge of adapting the algorithm to the boundedly uniform GRC
	model lies in the fact that the nodes are unable to sample from a distribution
	whose parameters depend on $n$.
	Nevertheless, we present a sampling procedure that allows each node
	$v \in V$
	to sample $\delta_{v}$ from a distribution that is \emph{sufficiently close}
	to the aforementioned capped geometric distribution.
	
	As we prove in \Sect{}~\ref{sec:spanner}, the sampling procedure allows us to
	construct a spanner with nearly the same properties as those of
	\cite{ForsterGV21}.
	More concretely, we extend and adapt the analysis of \cite{ForsterGV21} to
	show that our algorithm constructs a spanner with stretch
	$2 \kappa - 1$
	whp,
	and size
	$O (n^{1 + (1 + \varepsilon) / \kappa})$
	in expectation,
	where
	$\varepsilon > 0$
	is a constant parameter that can be made desirably small.
	
	\subparagraph{Lower Bounds.}
	Since the GRC model is subject to bandwidth constraints, with each pin carrying
	at most one bit of information per round, we wish to utilize the popular
	two-party communication complexity reduction framework, developed originally
	for CONGEST runtime lower bounds, in order to establish GRC runtime lower
	bounds.
	This framework is based on a partition of the node set of a carefully designed
	graph $G$ into (disjoint) sets $A$ and $B$, simulated by Alice and Bob,
	respectively.
	To adapt this framework to the GRC model, we aim to bound (from above) the
	number of bits that Alice and Bob need to exchange in order to simulate a
	round of a GRC algorithm over the graph $G$.
	
	Let us first consider the following naive communication scheme:
	for each pin associated with an edge in the
	$(A, B)$-cut,
	Alice (resp., Bob) sends a single bit that reflects whether a beep was
	transmitted on that pin from her (resp., his) side.
	Unfortunately, this scheme fails to truthfully simulate a round of the
	algorithm:
	Recalling the discussion in
	footnote~\ref{footnote:far-away-local-pin-partitions}, two pins
	$p, p'$
	associated with edges incident on nodes in $A$ (resp., $B$) may belong to the
	same circuit due to the local pin partitions of the nodes in $B$ (resp., $A$).
	In this case, Alice (resp., Bob) may not be able to determine whether $p$ and
	$p'$ belong to the same circuit and therefore, cannot simulate the behavior of
	their incident nodes.
	In \Sect{}~\ref{sec:lower-bounds}, we present a communication scheme that
	overcomes this obstacle while incurring a communication overhead which is only
	logarithmic in the size of the
	$(A, B)$-cut.
	
	It is important to note that our GRC runtime lower bounds can only use
	reductions that admit a ``static node partition'' structure.
	While these make for a rich class of reductions, one may wonder whether our
	lower bounds can be extended to reductions of a more dynamic structure,
	including, e.g., the reductions developed by Das Sarma et al.\ in
	\cite{dasHKKNPPW2011}.
	Our GRC runtime upper bounds demonstrate that this is not the case
	and in \Sect{}~\ref{sec:inapplicable-reductions}, we identify the ``point of
	failure'' that makes these reductions inapplicable to the GRC model.
	
	\section{Preliminaries}
	\label{sec:preliminaries}
	
	\subparagraph{Graph Theoretic Definitions.}
	Consider a connected graph
	$G = (V, E)$.
	Given an edge-weight function
	$w : E \rightarrow \Reals$,
	a \emph{minimum spanning tree (MST)} of $G$ with respect to $w$ is an edge
	subset
	$T \subseteq E$
	such that
	$(V, T)$
	is a spanning tree of $G$ that minimizes the weight
	$w(T) = \sum_{e \in T} w(e)$.
	
	For an edge subset
	$H \subseteq E$,
	let
	$d_{H}(u, v)$
	denote the distance in the graph
	$(V, H)$
	between two nodes
	$u, v \in V$.
	For an integer
	$\sigma > 0$,
	we say that
	$H \subseteq E$
	is a \emph{$\sigma$-spanner} of $G$ if
	$d_{H}(u, v) \leq \sigma \cdot d_{G}(u, v)$
	for all
	$u, v \in V$.
	Equivalently, $H$ is a $\sigma$-spanner if and only if
	$d_{H}(u, v) \leq \sigma$
	for every edge
	$(u, v) \in E$.
	The \emph{stretch} of $H$ is defined as the smallest value $\sigma$ for which
	$H$ is a $\sigma$-spanner.
	
	The parts of a partition $\mathcal{P}$ of the node set $V$ are often referred
	to as \emph{clusters}.
	We say that clusters $S$ and $S'$,
	$S \neq S'$,
	are \emph{neighboring clusters} if there exists an edge
	$(v, v') \in E$
	such that
	$v \in S$
	and
	$v' \in S'$.
	In this case, we say that edge
	$(v, v')$
	\emph{bridges} the clusters $S$ and $S'$, and more broadly, refer to
	$(v, v')$
	as a \emph{bridging} edge of $\mathcal{P}$.
	We say that an edge $(u, v) \in E$ is an \emph{outgoing} edge of cluster $S$
	if
	$u \in S$
	and
	$v \notin S$.
	For a cluster $S$, let
	$\partial_{S} \subseteq E$
	denote the set of edges outgoing from $S$.
	
	\subparagraph{Capped Geometric Distribution.}
	For parameters
	$p \in [0, 1]$
	and
	$r \in \mathbb{Z}_{> 0}$,
	the \emph{capped geometric distribution}, denoted by
	$GeomCap(p, r)$,
	is defined by taking
	$\Pr[GeomCap(p, r) = i]$
	to be
	$p (1 - p)^{i}$
	if
	$i \in \{ 0, \dots, r - 1 \}$;
	$(1-p)^{r}$
	if
	$i = r$;
	and $0$ otherwise.
	Intuitively, the distribution relates to $r$ Bernoulli experiments indexed by $0, \dots, r-1$, each with success probability $p$. A random variable sampled from the capped geometric distribution represents the index of the first successful experiment, whereas it is equal to $r$ if all experiments fail. The capped geometric distribution admits a memoryless property for the values
	$0\leq i \leq r - 1$.
	In particular, a useful identity that follows is $\Pr[X = i \mid X \geq i] = \Pr[X = 0] = p$ for a random variable $X\sim GeomCap(p,r)$ and an index  $0\leq i \leq r - 1$.

	\subsection{Auxiliary Procedures}\label{sec:auxiliary}
	
	\subparagraph{Global Circuits.}
	The algorithms presented in this paper utilize a \emph{global circuit}, i.e., a circuit in which every node $v\in V$ partakes.
	A global circuit can be constructed in round $t\geq 0$ as follows.
	For some index $1 \leq i \leq k$, every node $v \in V$
	partitions its pin set in round $t$ such that
	$E(v)\times \{i\}\in \PinPart^{t}(v)$.

	\subparagraph{Procedure $\mathtt{CountingToLogn}$.}
	We next present a procedure referred to as $\mathtt{CountingToLogn}$, whose runtime is $\Theta (\log n)$ rounds whp. 
	While the uniformity in $n$ prevents the nodes from counting $\log n$ rounds individually, the duration of this procedure can indicate to the nodes that whp, $\Theta (\log n)$ rounds have passed.
	The nodes first construct a global circuit, as described above.
	Throughout the procedure, the nodes maintain a node set $M \subseteq V$ of \emph{competitors}, where initially $M = V$.
	In each round, each competitor $v \in M$ tosses a fair coin and beeps through the global circuit if the coin lands heads.
	If the coin lands tails, $v$ removes itself from $M$.
	The procedure terminates when no competitor beeps through the global circuit. 
	
	We show the following useful property regarding the runtime of the described
	procedure.
	(All proofs missing from this section are deferred to
	Appendix~\ref{appendix:missing-proofs:preliminaries}).

	\begin{lemma}
		\label{lem:counting}
		For an integer $r > 0$, consider $2r - 1$ independent executions of $\mathtt{CountingToLogn}$ and let $\tau$ be the median runtime of these executions (i.e., the $r$-th fastest runtime).
		For any constant $0 < \rho < 1$, it holds that
		$\Pr[(1 - \rho) \log n \leq \tau \leq (1 + \rho) \log n]
		\geq
		1 - 2n^{-\rho r}$.
	\end{lemma}

	\subparagraph{Simulating a Message-Passing Network.}
	In a \emph{message-passing} network, in each round, every pair of neighboring
	nodes may exchange single bit messages with each other (cf.\ the
	$\text{CONGEST}(1)$
	model~\cite{peleg2000}).
	One can simulate a message-passing network in the GRC model using relatively
	standard techniques as cast in the following theorem.
	
	\begin{theorem}
		\label{thm:message-passing-simulation}
		Let \Alg{} be a GRC algorithm where additionally, in each round, each node is
		able to exchange $1$-bit messages with its neighbors.
		If the runtime of \Alg{} is $T$, then it can be transformed into an algorithm
		$\Alg'$ in the GRC model (without messages between neighbors) with a runtime
		of
		$O (\log n) + 4 T$
		whp.
	\end{theorem}
	
	For simplicity of presentation, we subsequently utilize
	\Thm{}~\ref{thm:message-passing-simulation} and describe our algorithms as if
	the nodes can exchange $1$-bit messages with their neighbors in each round.
	
	\subparagraph{Leader Election.}
	In the leader election task, the goal is for a single node in a given node set $I \subseteq V$ to be selected as a \emph{leader}, whereas all other nodes of $I$ are selected to be \emph{non-leaders}. Leader election is used as a procedure in some of our algorithms.
	To that end, we use a leader election algorithm presented by Feldmann et al.~\cite{feldmannPSD2022} in the context of reconfigurable circuits in the geometric amoebot model. We note that this leader election algorithm only uses a global circuit (as described above) and thus can be applied as-is in the GRC model. Hence, the following theorem is established.
	
	\begin{theorem}[\cite{feldmannPSD2022}]
		\label{thm:leader-election}
		The leader election task can be solved within $O(\log n)$ rounds whp.
	\end{theorem}
	
	\subparagraph{Outgoing Edge Detection.}
	Consider a graph $G=(V, E)$ and let $H\subseteq E$ be a subset of edges such that each node $v\in V$ knows the set of incident edges $H(v)$. Define a partition of $V$ into clusters according to the connected components of $(V,H)$. The objective of this procedure is for each node $v\in V$ to determine for each neighbor $u\in N(v)$, whether $u$ belongs to the same cluster as $v$. To that end, the nodes first construct a circuit for each cluster. This is done by each node $v\in V$ including the pin subset $H(v)\times \{i\}$ as part of its local pin partition for some $i\in [k]$ ($i$ is the same for all nodes). Then, each cluster elects a leader utilizing the leader election algorithm mentioned above. 
	The selected leader of each cluster tosses $\Theta (\log n)$ bits and beeps them through the cluster's circuit, one at a time (a beep represents $1$ and silence represents $0$).
	Since the nodes cannot count $\Theta (\log n)$ rounds, \Proc{}~$\mathtt{CountingToLogn}$ is executed in parallel through a global circuit for (a sufficiently large) $c > 1$ times, indicating to the clusters' leaders how long to continue with the bit tossing process. 
	Every node $v \in V$ sends every bit received through its cluster's circuit in a direct message to all its neighbors (messages between neighbors are executed by means of the simulation method described in \Sect{}~\ref{sec:auxiliary}).
	For every incident edge $e \in E(v)$, node $v$ checks if the bit received differs from the bit sent. If so, $e$ is classified by $v$ as an outgoing edge.
	
	\begin{lemma}
		\label{lem:outgoing}
		In the outgoing edge detection procedure, every edge
		$e = (u, v) \in E$
		is classified correctly whp by both $u$ and $v$.
	\end{lemma}
	
	\begin{lemma}
		\label{lem:outgoing-runtime}
		The outgoing edge detection procedure takes $\Theta (\log n)$ rounds whp.
	\end{lemma}
	
	\section{A Fast Minimum Spanning Tree Algorithm}
	\label{sec:mst}
	In this section, we present a randomized MST algorithm that operates in the GRC model.
	As common in the distributed setting, we assume the edge-weights are integers from the set $\{ 1, \dots, W \}$ for some positive integer $W$.
	Each node $v\in V$ initially knows only the weights of edges in $E(v)$.
	In particular, as dictated by the GRC model, node $v$ does not know the value of $W$ or any other information about $W$.
	
	Our algorithm can be seen as an adaptation of Boruvka's classical MST algorithm~\cite{boruvka1926} to the GRC model.
	Throughout its execution, Boruvka's algorithm maintains an edge set $T$ and a cluster partition defined such that each cluster is a connected component of $(V, T)$.
	Initially, $T = \emptyset$ (and each node is a cluster).
	At each iteration of the algorithm, each cluster $S$ adds a lightest outgoing edge
	$e^{*} = \arg\min_{e \in \partial_{S}} \{ w(e) \}$
	to $T$.
	This means that $S$ merges with the neighboring cluster $S'$ that is incident on $e^{*}$.
	It is well-known that if the edge weights are unique, then Boruvka's algorithm computes an MST of $G$.
	Notice that in our case, edge weights are not necessarily unique, so we construct a symmetry-breaking mechanism based on a total order of the lightest outgoing edges as explained later on. 
	
	The algorithm begins with an empty set of \emph{tree edges} and operates in phases.
	The goal of each phase is to add tree edges similarly to Boruvka's algorithm.
	Let
	$T_{i} \subseteq E$
	denote the tree edges at the end of phase $i\geq 0$.
	As in Boruvka's algorithm, the connected components of $(V, T_{i})$ are defined to be the \emph{clusters} at the beginning of phase $i + 1$.
	The nodes construct a designated circuit for each cluster formed during the algorithm. Additionally, the nodes communicate through a global circuit and exchange messages with their neighbors using the methods described in \Sect~\ref{sec:preliminaries}.
	The operation of each phase is divided into the following stages.
	
	
	\subparagraph{Outgoing Edge Detection.}
	The purpose of this stage is to allow the nodes to identify which of their incident edges is an outgoing edge.
	To that end, the nodes execute the outgoing edge detection procedure described in \Sect{}~\ref{sec:auxiliary}.
	When the procedure terminates, each node detecting an outgoing edge beeps through the global circuit.
	The algorithm terminates if no node beeps in this round through the global circuit.
	Otherwise, the nodes advance to the next stage.
	Denote by $\mathtt{Out}(v)$ the set of edges classified as outgoing by node $v\in V$.
	
	\subparagraph{Lightest Edge Detection.}
	In this stage, each cluster searches for its lightest outgoing edges.
	Fix some cluster $S$.
	At the beginning of this stage, every node $v \in S$ such that
	$\mathtt{Out}(v) \neq \emptyset$
	marks a single edge
	$e \in \mathtt{Out}(v)$
	with weight
	$w(e) = \min_{e'\in \mathtt{Out}(v)} w(e')$
	as a candidate.
	The comparison between weights of the candidate edges incident on the nodes of $S$ is done in two steps. 
	
	First, the nodes compare the lengths of the candidate edge weights (i.e., the number of bits in the edge-weight representation).
	Consider a node $v \in S$ incident on a candidate edge $e$, and let $\ell_{v}=\lfloor \log w(e)\rfloor+1$ be the length of $w(e)$.
	Node $v$ counts $\ell_{v}-1$ rounds.
	If $v$ hears a beep on the cluster's circuit during those $\ell_{v}-1$ rounds, then $v$ unmarks $e$ as a candidate.
	Otherwise, $v$ beeps through the cluster's circuit in round $\ell_{v}$ and keeps $e$ as a candidate edge.
	Following the first step, all remaining candidate edges of $S$ have weights of the same length.
	In the second step, the weights of the candidate edges of $S$ are compared bit by bit, starting from the most significant bit.
	Let $v \in S$ be a node that still has an incident candidate edge $e$.
	The second step runs for $\ell_{v}$ rounds indexed by
	$j = 1, \dots, \ell_{v}$.
	In round $j$, if $e$ is still a candidate, then $v$ beeps through the cluster's circuit if and only if the $j$-th most significant bit of $w(e)$ is $0$.
	If $v$ did not beep but heard a beep through the cluster's circuit, it unmarks $e$ as a candidate edge.
	Notice that at the end of the second step, only the lightest edges that were classified as outgoing remain candidates. 

	In parallel, $v$ beeps through the global circuit at every round of the stage in which $e$ is still a candidate.
	Once $v$ finishes the stage (either because $e$ was marked as a lightest outgoing edge or $e$ was unmarked as a candidate), it stops beeping through the global circuit.
	The stage terminates when no beep is transmitted through the global circuit.   
	
	\subparagraph{Single Edge Selection.}
	At this point, only the edges marked as lightest outgoing edges of each cluster remained candidates.
	However, there may be more than one candidate edge for some clusters.
	The goal of this stage is to select a single edge for each cluster while avoiding the formation of a cycle in the output edge set (as we will show in the analysis).
	To that end, every node $v \in V$ with an incident candidate edge $(u, v)$ informs $u$ that $(u, v)$ is still a candidate.
	Then, each of $u$ and $v$ draws a random bit denoted by $u.bit$ and $v.bit$, respectively.
	Node $u$ sends $u.bit$ to $v$ and $v$ calculates the bitwise XOR of $u.bit$ and $v.bit$.
	Node $v$ beeps through the cluster's circuit if the XOR result is $1$.
	If node $v$ does not beep for edge $e$ but hears a beep through the cluster's circuit, it unmarks $e$ as a candidate. 
	Notice that if $(u, v)$ is lightest with regard to $u$'s cluster as well, then the same operation is performed also by $u$ using the same drawn bits.
	This edge selection process is done in parallel to \Proc{}~$\mathtt{CountingToLogn}$ over the global circuit, executed (a sufficiently large) $c >1$ times.
	The nodes continue to draw bits for their incident candidate edges as long as \Proc{}~$\mathtt{CountingToLogn}$ continues.
	If a node $v\in V$ has an incident candidate edge $e=(u,v)$ at the end of this stage, then it informs $u$, and both endpoints mark $e$ as a tree edge. 
	
	\subparagraph{Updating the Local Pin Partition.}
	Every node $v\in V$ sets its local pin partition to include the pin subset
	$T(v) \times \{ j \}$ for some $j\in [k]$,
	where $T(v)$ is the set of edges incident on $v$ that were marked as tree edges (either in the current or a prior phase).
	Observe that this local pin partition by the nodes constructs a circuit for every cluster.
	
	{\vskip 1pc}
	The output of the algorithm is the set of all tree edges.
	
	\subsection{Analysis}
	\label{sec:mst-analysis}
	In this section, we prove the correctness and analyze the runtime of the MST algorithm presented above, establishing the following theorem.
	
	\begin{theorem}
		\label{thm:mst}
		The algorithm constructs an MST of $G$ whp and runs in $O (\log n \cdot \log (n + W))$
		rounds whp.
	\end{theorem}
	
	Recall that
	$T_{i} \subseteq E$
	is the set of tree edges at the end of phase
	$i = 0, 1, \dots$
	and let $i^{*}$ be the last phase of the algorithm.
	Let $q_{i}$ be the number of clusters maintained by the algorithm at the
	beginning of phase $i$, that is, the number of connected components in
	$(V, T_{i})$.
	
	\begin{lemma}
		\label{lem:mst-connected-components}
		Consider a phase $0 \leq i \leq i^{*}$. If $q_{i} = 1$, then the algorithm terminates in phase $i$ whp; otherwise,
		$q_{i + 1} \leq \frac{1}{2} q_{i}$
		whp.
	\end{lemma}
	
	The proof of \Lem{}~\ref{lem:mst-connected-components} is deferred to
	Appendix~\ref{appendix:missing-proofs:mst}.
	Notice that since the algorithm starts with $n$ clusters,
	\Lem{}~\ref{lem:mst-connected-components} implies the following corollary.
	
	\begin{corollary}
		\label{cor:mst-termination-phases}
		The algorithm terminates after $i^{*} = O (\log n)$ phases whp. Moreover, the subgraph $(V, T_{i^{*}})$
		is connected whp.
	\end{corollary}

	Denote by $D_{i} \subseteq E$ the set of edges that are candidates for some (at least one) cluster at the end of the single edge selection stage of phase $i$ (to be marked as tree edges).
	\begin{lemma}
		\label{lem:mst-spanning-tree}
		The subgraph $(V, T_{i^{*}})$ is a spanning tree of $G$ whp.
	\end{lemma}
	\begin{proof}
		By \Cor{}~\ref{cor:mst-termination-phases}, $(V, T_{i^{*}})$ is connected whp. So, it is left to show that $(V, T_{i^{*}})$ is a forest whp. We prove by induction over the phases that $(V, T_{i})$
		is a forest whp for all $0 \leq i \leq i^{*}$.
		\Cor{}~\ref{cor:mst-termination-phases} also guarantees that there are $O (\log n)$ phases whp; hence the statement follows by applying union bound over the phases.
		
		For the base of the induction, notice that $T_{0} = \emptyset$, and thus $(V, T_{0})$ is a forest. Now, suppose that $(V, T_{i})$ is a forest for some $0 \leq i < i^{*}$.
		We show that
		$(V, T_{i + 1})
		=
		(V, T_{i} \cup D_{i})$
		is a forest whp. For every edge $e \in D_{i}$, let $B_{i}(e)$ be the integer obtained from the binary representation of the bit sequence drawn for $e$ by its endpoints (i.e., the sequence of XORed bits) in the single edge selection stage of phase $i$.
		Define the binary relation $\prec_{i}$ for every two edges $e, e' \in D_{i}$ as: 
		\[
		e \prec_{i}	e'
		\iff
		w(e) < w(e')
		\lor
		\left( w(e) = w(e') \land B_{i}(e) > B_{i}(e') \right) \, .
		\]
		Notice that by repeating the
		$\mathtt{CountingToLogn}$
		for a sufficiently large number of times, we get that the $B_{i}(\cdot)$ values are unique whp. By the construction of the single edge selection stage, this means that each cluster selects exactly one outgoing edge whp --- the lightest outgoing edge which is minimal with respect to $\prec_{i}$.
		To complete our proof, we show that if the $B_{i}(\cdot)$ values are unique and $(V, T_{i})$ is a forest, then
		$(V, T_{i + 1}) = (V, T_{i}\cup D_{i})$
		is a forest. 
		
		Assume by contradiction that there exists at least one cycle in $(V, T_{i}\cup D_{i})$ and let $Y$ be a simple cycle in $(V, T_{i} \cup D_{i})$.
		By the induction hypothesis we know that $(V, T_{i})$ is a forest, therefore
		$Y \cap D_{i} \neq \emptyset$.
		Let $e\in Y\cap D_{i}$ be the (unique) largest edge (with respect to $\prec_{i}$) of $Y \cap D_{i}$, and let $S$ be the cluster that selected $e$.
		Observe that since $(V, T_{i})$ is a forest and $Y$ is a cycle, there exists another edge $e'\in Y\cap D_{i} - \{ e \}$
		which is an outgoing edge of $S$.
		However, by the choice of $e$, we know that $e'\prec_{i} e$,
		in contradiction to the selection of $e$ by $S$.
	\end{proof}
	
	The following lemma asserts the correctness of our MST algorithm.
	
	\begin{lemma}
		\label{lem:mst}
		The graph $(V, T_{i^{*}})$ is an MST of $G$ whp.
	\end{lemma}
	\begin{proof}
		By \Lem{}~\ref{lem:mst-spanning-tree}, the graph $(V, T_{i^{*}})$ is a spanning tree of $G$ whp.
		The proof of \Lem{}~\ref{lem:mst-spanning-tree} shows that every cluster selects a single lightest outgoing edge in each phase whp.
		The statement then follows from the correctness of Boruvka's algorithm~\cite{boruvka1926}.
	\end{proof}
	
	It remains to analyze the runtime of the algorithm.
	
	\begin{lemma}
		\label{lem:mst-runtime}
		The MST algorithm runs in
		$O (\log n \cdot \log (n + W))$
		rounds whp.
	\end{lemma}
	\begin{proof}
		By Corollary \ref{cor:mst-termination-phases}, the algorithm runs for $O (\log n)$ phases whp.
		We are left to bound the runtime of each phase.
		Every execution of the leader election algorithm and \Proc{}~$\mathtt{CountingToLogn}$ takes $O (\log n)$ rounds whp.
		Hence, the outgoing edge detection and single edge selection stages each take $O (\log n)$ rounds whp.
		The lightest edge detection stage completes in $O(\log W)$ rounds, and updating the local pin partition does not require any communication.
		Therefore, every phase of the algorithm completes in $O (\log n + \log W)$ rounds whp. Overall, we get a runtime bound of $O (\log n (\log n + \log W))=O (\log n \cdot \log (n \cdot W))=O(\log n \cdot \log (n + W))$
		rounds whp, where the last equality hods because $\log (n \cdot W)=\log n+\log W=O(\log (n+W) )$.
	\end{proof}
	\section{A Sparse Spanner Algorithm}
	\label{sec:spanner}
	In this section, we present a randomized spanner algorithm that operates in
	the GRC model.
	Given a parameter $\kappa\in \mathbb{Z}_{>0}$ and a constant $0 < \varepsilon
	< 1$, the algorithm constructs a spanner with a stretch of $(2 \kappa - 1)$
	whp and $O (n^{1 + (1 + \varepsilon) / \kappa})$ edges in expectation.
	More concretely, we prove the following theorem.
	
	\begin{theorem}\label{thm:spanner}
		There exists an algorithm in the GRC model that computes a set $H\subseteq E$
		of edges such that $H$ is a $(2 \kappa - 1)$-spanner whp, and $\mathbb{E}[|H|]
		= O (n^{1 + \frac{1 + \varepsilon}{\kappa}})$.
		The runtime of the algorithm is $O (\kappa+\log n)$ rounds whp, and the memory
		space used by each node $v \in V$ is $O (\deg (v) + \kappa)$.
	\end{theorem}
	
	In \Sect{}~\ref{sec:spanner-adaptation-to-small-memory}, we present a
	modification of our algorithm to accommodate a memory space of only
	$O (\deg(v) + \log \kappa)$
	for each node $v \in V$, at the cost of a slightly
	slower $O (\kappa \log n)$-round algorithm.
	We start by describing the algorithm stated in \Thm{}~\ref{thm:spanner}.
	
	The algorithm is based on the random shift concept introduced by Miller et al.\ in \cite{MillerPX13} and studied further in various works (see, 
	e.g.,~\cite{
		MillerPVX15,
		ElkinN19,
		ForsterGV21}).
	We now give a high-level overview of a spanner construction algorithm based on the random shift approach (see \cite{ForsterGV21} for the full details).
	
	The algorithm starts with each node $v \in V$ sampling a value $\delta_{v}\sim GeomCap(1 - n^{-1 / \kappa}, \kappa - 1)$ (see \Sect~\ref{sec:preliminaries} for the capped geometric distribution definition). Then, the nodes conceptually add a virtual node $s$. Each node $v \in V$ adds an edge $(s, v)$ of weight
	$w(s, v) = \kappa - \delta_{v}$ to form the graph $G'$, where all other edges are assigned a unit weight.
	Following that, the nodes construct a shortest path tree $T$ rooted at $s$.
	The nodes of $G$ are partitioned into clusters defined by the connected components of $T$ after removing $s$ and its incident edges.
	To construct the spanner $H$, the nodes first add the (non-virtual) edges of $T$.
	Then, the nodes add edges to $H$ such that for each edge $(u, v) \in E-T$, at least one of the following is satisfied: (1) $H$ contains exactly one edge between $u$ and a node in $v$'s cluster; or (2) $H$ contains exactly one edge between $v$ and a node in $u$'s cluster. As discussed in \cite{ForsterGV21}, the constructed edge-set $H$ is a $(2 \kappa - 1)$-spanner of expected size $O (n^{1+1/\kappa})$.
	
	Our algorithm works in three stages as described below.
	
	\subparagraph{Sampling Procedure.}
	Recall that the algorithm of \cite{ForsterGV21} begins with each node $v \in V$ sampling $\delta_{v}\sim GeomCap(1 - n^{-1 / \kappa}, \kappa - 1)$. Note that sampling from
	$GeomCap(1 - n^{-1 / \kappa}, \kappa - 1)$
	requires the nodes to know the value of $n$, which is not possible in the GRC model. Hence, we devise a designated sampling procedure for each node $v \in V$. 
	
	Let us first present the intuition behind the sampling procedure. The idea is for each node $v \in V$ to simulate $\kappa - 1$ \emph{experiments}, each with success probability close to $1 - n^{-1 / \kappa}$, and compute $\delta_{v}$ accordingly. To achieve such success probability without knowing $n$, \Proc{}~$\mathtt{CountingToLogn}$ is utilized.
	In order to enhance the proximity to $1 - n^{-1 / \kappa}$, \Proc{}~$\mathtt{CountingToLogn}$ is executed numerous times in parallel, and $\delta_{v}$ is computed based on the run with median runtime.
	
	For ease of presentation, we describe the sampling procedure in two stages. First, a sub-procedure referred to as the \emph{basic} scheme is described. We later explain how this basic scheme is used in the sampling procedure. The basic scheme runs during an execution of \Proc{}~$\mathtt{CountingToLogn}$. For each node $v \in V$, let
	$b_{v} = (b_{v}[0], \dots, b_{v}[\kappa - 2])$
	be a vector of $\kappa - 1$ bits initialized to $b_{v} =( 0, \dots, 0)$. The purpose of entry $b_{v}[j]$ is to represent the success/failure of the $i$-th experiment for each $0 \leq j \leq \kappa - 2$.
	Let $\varepsilon'$ be the largest value such that $1 / (1 - \varepsilon')$ is an integer and $\varepsilon'\leq \varepsilon/(2+\varepsilon)$. In each round $j$ such that $j\bmod \kappa\neq 0$, each node $v$ draws $1 / (1 - \varepsilon')$ bits uniformly at random and sets
	$b_{v}[(j-1) \bmod \kappa] = 1$
	if any of those bits are $1$.
	
	In the sampling procedure, the nodes perform
	$c' = 2\cdot \lceil c / \varepsilon '\rceil - 1$
	executions of the basic scheme, where $c>0$ is a constant.
	Let us index these executions by
	$i = 0, \dots, c' - 1$.
	Starting from the execution indexed $0$, the rounds of the executions are done alternately, i.e., a round of the run indexed by $i$ is followed by a round of the run indexed by $(i + 1) \bmod c'$. Accordingly, each node $v \in V$ maintains $c'$ vectors, $b_{v}^{0}, \dots, b_{v}^{c' - 1}$, each of size $\kappa - 1$ bits, such that $b_{v}^{i}$ is the vector maintained by $v$ during the $i$-th execution of the basic scheme. Additionally, $v$ maintains a counter initialized to $0$, whose goal is to count the executions that terminated.
	Whenever an execution terminates, the counter is increased by $1$.
	Following the termination, during the rounds that are associated with that execution, the nodes do nothing. The nodes halt the executions when the counter reaches $\lceil c / \varepsilon'\rceil$ (notice that the counter is updated in the same manner for all nodes, thus they halt at the same time). Let $\tilde{i}$ denote the index of the execution in which the counter reached $\lceil c / \varepsilon' \rceil$. Observe that this is the $\lceil c / \varepsilon' \rceil$-th fastest execution, i.e., the execution with median runtime.
	Each node $v \in V$ defines $\delta_{v}$ to be the smallest index $0 \leq j\leq \kappa - 2$ for which $b_{v}^{\tilde{i}}[j] = 1$ if such an index exists, or $\delta_{v} = \kappa - 1$ otherwise.
	
	\subparagraph{Partition Into Clusters.}
	Let $G'$ be the graph formed by adding a virtual node $s$ and edge $(s, v)$ of weight
	$w(s, v) = \kappa - \delta_{v}$
	for every $v \in V$.
	To compute the cluster partition, the nodes first construct a shortest path tree $T$ rooted at $s$. The idea is simple: If $w(s, v) = 1$, then $v$ sends a message to all its neighbors and marks itself as the center of its cluster. 
	Otherwise, assume first that $v$ receives a message in at least one of the rounds
	$2, \dots, w(s, v) - 1$
	and let
	$2 \leq i < w(s, v) - 1$
	be the first such round.
	After receiving a message in round $i$, node $v$ (arbitrarily) chooses a neighbor $u$ that sent $v$ a message in that round and adds the edge $(u, v)$ into $T$. Then, in round $i + 1$, node $v$ sends a message to all neighbors from which it did not receive a message in round $i$.
	Otherwise, if $v$ does not receive a message after $w(s, v) - 1$ rounds, then in round $w(s, v)$ node $v$ sends a message to all its neighbors and sets itself as the center of its cluster.
	Notice that after at most $\kappa$ rounds, $T$ is a shortest path tree rooted at $s$.
	The edges of $T$ are added to the spanner $H$. The clusters are defined to be the connected components of $(V, T)$ (i.e., the connected components formed by removing $s$ and its incident edges).
	The nodes then construct a circuit for each cluster (similarly to the MST algorithm of \Sect{}~\ref{sec:mst}). Observe that by design, each cluster has exactly one center. 
	Note that every message sent in each round of this stage is of size one bit.
	
	\subparagraph{Addition of Bridging Edges.}
	The construction of $H$ is completed by the following procedure whose goal is to augment $H$ with some of the edges that bridge between clusters.
	This is done by each cluster randomly drawing an ID.
	Then, each node $v \in V$ identifies its neighboring clusters with smaller IDs and adds a single edge to each such cluster into $H$. 
	
	Formally, each node $v \in V$ maintains a set $S^{\mathrm{eq}}(v)$ initialized to be $N(v)$, and a set $S^{\mathrm{sml}}(v)$ initialized to be $\emptyset$. Additionally, throughout the execution, $v$ maintains a partition of $S^{\mathrm{sml}}(v)$ into subsets according to the (randomly drawn) cluster IDs.
	The nodes engage in a process that runs in parallel to $4c + 7$ iterations of \Proc{}~$\mathtt{CountingToLogn}$. In each round of this process, every cluster center tosses a coin and communicates the outcome through the cluster's circuit to all the nodes in its cluster. Then, every node $v \in V$ sends a message with the coin toss received from its cluster's center to all neighbors. Let $S^{\mathrm{eq}}_{i}(v)$ be the set $S^{\mathrm{eq}}(v)$ at the beginning of round $i$. For each $u \in S^{\mathrm{eq}}_{i}(v)$, if $u$ and $v$ sent the same bit, then $u$ stays in $S^{\mathrm{eq}}(v)$; otherwise, $u$ is removed. Additionally, if $u$'s bit is smaller than $v$'s, then $u$ is added to $S^{\mathrm{sml}}(v)$.
	The partition of the nodes in $S^{\mathrm{sml}}(v)$ is defined so that $u$ and $u'$ are in the same subset by the end of round $i$ if and only if they were in the same subset at the beginning of round $i$ and sent the same bit in round $i$.
	Let
	$s_{1}, \dots, s_{q}$
	be the partition of $S^{\mathrm{sml}}(v)$ at the end of the process.
	For each $j\in [q]$, node $v$ (arbitrarily) selects a single node $u \in s_{j}$ and adds the edge $(u, v)$ into $H$.
	
	This completes the construction of $H$. 
	We now turn to analyzing the algorithm.
	
	\subsection{Analysis}\label{sec:spanner-analysis}
	
	This section is dedicated to proving \Thm{}~\ref{thm:spanner}.
	To that end, we start with a structural lemma about the capped geometric
	distribution.
	(All proofs missing from this section are deferred to
	Appendix~\ref{appendix:missing-proofs:spanner}).
	
	\begin{lemma}\label{lem:spanner-capped}
		For arbitrary values
		$q_{1}, \dots, q_{n}$
		and for
		$X_{1}, \dots, X_{n}\sim GeomCap(\phi, \kappa - 1)$,
		define
		$M = \max _{i \in [n]} \{ X_{i} - q_{i} \}$.
		For the set
		$I = \{ i \mid X_{i} < \kappa - 1\ \land\ X_{i} - q_{i} \in \{ M - 1, M \} \}$,
		it holds that
		$\mathbb{E}[|I|] \leq \frac{2}{1 - \phi}$. 
	\end{lemma}
	
	Recall that in the sampling procedure of our algorithm, the value $\delta_{v}$ is computed for each node $v \in V$ based on the $\tilde{i}$-th execution of the basic scheme, i.e., the execution that admits the median runtime.
	Particularly, within that execution, $\delta_{v}$ is defined as the first successful experiment out of $0, \dots, \kappa - 2$; or $\kappa - 1$ if all experiments failed. Let $\phi$ be the success probability of each such experiment and notice that $\phi$ itself is a random variable that depends on the execution's length. Define $A$ to be the event that $1 - n^{-1 / \kappa} \leq \phi \leq 1 - n^{-(1 + \varepsilon) / \kappa}$.
	We prove the following lemma.
	
	\begin{lemma}
		\label{lem:spanner-A}
		$\Pr[A] \geq 1 - 2n^{-c}$.
	\end{lemma}
	\begin{proof}
		Let $\tau$ denote the length of execution $\tilde{i}$ in the sampling procedure.
		That is, the median runtime out of $2\cdot \lceil c / \varepsilon '\rceil - 1$ executions of \Proc{}~$\mathtt{CountingToLogn}$.
		By \Lem{}~\ref{lem:counting}, it follows that
		$\Pr[(1 - \varepsilon') \log n \leq \tau \leq (1 + \varepsilon') \log n] \geq 1 - 2n^{-\varepsilon' \cdot \lceil c / \varepsilon '\rceil} \geq 1 - 2n^{-c}$.
		Let $\sigma$ be the total number of random bits designated for each experiment of execution $\tilde{i}$. For each experiment, the number of rounds in which the nodes draw bits is $\tau/\kappa$.
		Since
		$1 / (1 - \varepsilon')$
		bits are drawn in every such round and since
		$\frac{1 + \varepsilon'}{1 - \varepsilon'}
		\leq
		\frac{1 + \varepsilon/(2+\varepsilon)}{1 - \varepsilon/(2+\varepsilon)} = 1 + \varepsilon$,
		we get that
		\begin{align*}
			&\Pr\left[ \frac{\log n}{\kappa} \leq \sigma \leq (1 + \varepsilon)\frac{\log n}{\kappa} \right]
			\geq 
			\Pr[(1 - \varepsilon') \log n \leq \tau \leq (1 + \varepsilon') \log n] \geq 1 - 2n^{-c}\ .\
		\end{align*}
		Observe that by design,
		$\phi = 1 - (1 / 2)^{\sigma}$
		and thus, it follows that
		$\Pr[A] \geq 1 - 2n^{-c}$. 
	\end{proof}
	
	We now consider the bridging edges addition stage of the algorithm.
	Let $B$ denote the event that for every edge $(u, v) \in E-T$, at least one of the following is satisfied:
	(1)
	$H$ contains exactly one edge between $u$ and a node in $v$'s cluster;
	or
	(2)
	$H$ contains exactly one edge between $v$ and a node in $u$'s cluster.
	We prove the following.
	
	\begin{lemma}
		\label{lem:spanner-B}
		$\Pr[B] \geq 1-3n^{-c}$.
	\end{lemma}
	
	For each node $v \in V$, let
	$M_{v} = \max_{u \in V} \{ \delta_{u} - d_{G}(u, v) \}$
	and
	$R(v) = \{u \in V \mid M_{v} - 1 \leq \delta_{u} - d_{G}(u, v) \leq M_{v}\}$.
	We obtain the following observation.
	
	\begin{observation}\label{obs:spanner-R(v)}
		Consider an edge $(u, v) \in H$ such that $u$ and $v$ belong to clusters centered at nodes $u'$ and $v'$, respectively.
		Then, $u'\in R(v)$ or $v'\in R(u)$.
	\end{observation}
	\begin{proof}
		First, observe that either 
		(1)
		$d_{G'}(s, v) \leq w(s, u') + d_{G}(u', v) \leq d_{G'}(s, v)+1$; or
		(2)
		$d_{G'}(s, u) \leq w(s, v') + d_{G}(v', u) \leq d_{G'}(s, u) + 1$ (or both).
		Assume w.l.o.g.\ that case (1) holds (the second case is analogous).
		Notice that
		\[
		w(s, u') + d_{G}(u', v)
		\, = \,
		\kappa-\delta_{u'} + d_{G}(u', v)
		\, = \,
		\kappa - (\delta_{u'} - d_{G}(u', v)) \,.
		\]
		Now, by the definition of distance, we have
		\begin{align*}
			&d_{G'}(s, v)
			\, = \,
			\min_{x\in V} \{w(s, x) + d_{G}(x, v)\}
			\, = \,\\
			&\min_{x \in V} \{\kappa - (\delta_{x} - d_{G}(x, v)) \}
			\, = \,
			\kappa - M_{v} \, .
		\end{align*}
		Overall, it follows that
		\begin{align*}
			& \kappa-M_{v}
			\, \leq \, \kappa - (\delta_{u'} - d_{G}(u', v))
			\, \leq \,
			\kappa - (M_{v} - 1) \\
			&\implies
			M_{v} - 1
			\, \leq \, \delta_{u'} - d_{G}(u', v)
			\, \leq \,
			M_{v}
			\, \implies \,
			u' \in R(v) \, .\qedhere
		\end{align*}
	\end{proof}
	
	We are now prepared to bound the expected number of edges in the spanner.
	\begin{lemma}\label{lem:spanner-edges}
		$\mathbb{E}[|H|] \leq 2n^{1 + (1 + \varepsilon) / \kappa}+n^{1+1/\kappa}+1$.
	\end{lemma}
	\begin{proof}
		By the law of total expectation,
		\[
		\mathbb{E}[|H|]
		\, = \,
		\mathbb{E}[|H| \mid A \land B]\cdot\Pr[A \land B] + \mathbb{E}[|H| \mid \lnot A \lor \lnot B]\cdot\Pr[\lnot A\lor \lnot B] \, .
		\]
		Combining \Lem{}~\ref{lem:spanner-A} with \Lem~\ref{lem:spanner-B}, we get
		$\Pr[\lnot A \lor \lnot B] \leq 5n^{-c}$,
		and since
		$\mathbb{E}[|H| \mid \lnot A \lor \lnot B] \leq m<n^2$, it follows that
		\[
		\mathbb{E}[|H|] \leq \mathbb{E}[|H| \mid A \land B] \cdot \Pr[A \land B]+n^{2} \cdot 5n^{-c}
		\leq
		\mathbb{E}[|H| \mid A \land B] + 1 \, ,
		\]
		where the final inequality holds for, e.g., $c\geq 3$.
		Therefore, we are left to bound the term
		$\mathbb{E}[|H| \mid A \land B]$. 
		
		\Obs{}~\ref{obs:spanner-R(v)} implies that the sum
		$\sum_{v \in V}|R(v)|$
		accounts for every edge in $H$ at least once, i.e., $\sum_{v \in V}|R(v)| \geq |H|$. Fix some node $v \in V$, we seek to bound $\mathbb{E}[|R(v)|]$. Partition the set $R(v)$ into $R_{1}(v)
		=
		\{u \in R(v)\mid \delta_{u}
		=
		\kappa - 1\}$
		and
		$R_{2}(v) = R(v)-R_{1}(v)$. 
		Notice that the events $\delta_{u} = \kappa - 1$ and $B$ are independent.
		Thus, we get 
		\[
		\mathbb{E}[|R_{1}(v)| \mid A \land B] \leq n \cdot \Pr[ \delta_{u} = \kappa - 1 \mid A \land B]
		\, = \,
		n \cdot \Pr[ \delta_{u} = \kappa - 1 \mid A] \, .
		\]
		Observe that
		$ \mathbb{E}[|R_{1}(v)|] \leq n \cdot \Pr[ \delta_{u} = \kappa - 1] = n(1 - \phi)^{\kappa - 1}$,
		and recall that if event $A$ occurs, then
		$\phi\geq 1 - n^{-1 / \kappa}$.
		Hence, it follows that
		\[
		n \cdot \Pr[ \delta_{u} = \kappa - 1 \mid A]
		\, = \,
		n(1 - \phi)^{\kappa - 1} \leq n \cdot n^{(-1 / \kappa) \cdot (\kappa - 1)}
		\, = \, n^{1 / \kappa} \, .
		\]
		As for $R_{2}$, applying \Lem{}~\ref{lem:spanner-capped}, we get $\mathbb{E}[|R_{2}(v)|] \leq2 / (1 - \phi)$.
		Once again, we condition on $A$ and $B$ to get
		\[
		\mathbb{E}[|R_{2}(v)| \mid A \land B]
		\, = \,
		\mathbb{E}[|R_{2}(v)| \mid A]
		\, \leq \,
		2/n^{-(1 + \varepsilon) / \kappa}
		\, = \,
		2n^{(1 + \varepsilon) / \kappa} \, .
		\]
		Overall, we conclude that
		\begin{align*}
			&\mathbb{E}[|H|]
			\, \leq \,
			n \cdot \mathbb{E}[R(v)]
			\, \leq \,
			n \cdot \mathbb{E}[|R_{1}(v)| \mid A \land B] + n \cdot \mathbb{E}[|R_{2}(v)| \mid A \land B] + 1 \\
			&\, \leq \,
			2n^{1 + (1 + \varepsilon) / \kappa} + n^{1 + 1/\kappa} + 1 \, .\qedhere
		\end{align*}
	\end{proof}
	Next, we bound the stretch of $H$.
	\begin{lemma}\label{lem:spanner-stretch}
		$H$ is a $(2 \kappa - 1)$-spanner whp.
	\end{lemma}
	\begin{proof}
		We now argue that if event $B$ occurs, then $H$ has stretch $2 \kappa - 1$, which implies the stated claim due to \Lem{}~\ref{lem:spanner-B}.
		To see that, consider an edge $(u, v) \in E$.
		Observe that the diameter within each cluster is at most $2 \kappa - 2$.
		This is because every node is at distance at most $\kappa - 1$ from its cluster's center.
		Hence, if $u$ and $v$ belong to the same cluster, then
		$d_{H}(u, v) \leq 2 \kappa - 2$.
		Otherwise, if event $B$ occurs, then either there is an edge $(\tilde{u}, v) \in H$ between $v$ and a node $\tilde{u}$ in $u$'s cluster, or an edge $(u, \tilde{v}) \in H$ between $u$ and a node $\tilde{v}$ in $v$'s cluster.
		Assume w.l.o.g.\ that $(\tilde{u}, v) \in H$.
		It follows that
		$d_{H}(u, v) \leq d_{H}(u, \tilde{u}) + d_{H}(\tilde{u}, v) \leq 1 + 2 \kappa - 2 = 2 \kappa - 1$.
	\end{proof}
	We conclude the analysis of our algorithm with the proof of \Thm{}~\ref{thm:spanner}.
	\begin{proof}[Proof of \Thm{}~\ref{thm:spanner}]
		The correctness of the algorithm follows from \Lem{}~\ref{lem:spanner-edges} and \Lem{}~\ref{lem:spanner-stretch}.
		For the runtime, observe that the first and third stages take $O (\log n)$ whp, and the second stage takes $O (\kappa)$ rounds since the depth of the shortest path tree $T$ is at most $\kappa$.
		Regarding the memory space used by each node $v \in V$, the sampling procedure requires a constant number of $O (\kappa)$-sized vectors, along with a constant number of $O (\log \kappa)$-sized counters (to associate each round with the corresponding experiment).
		The partition into clusters requires a memory space of $O (\log \kappa)$ (maintaining a counter for the round number), and the addition of bridging edges requires $O (\deg(v))$ memory space.
		Overall, the memory space used is $O (\deg(v) + \kappa)$. 
	\end{proof}

	\subsection{Adaptation to Small Memory Space} 
	\label{sec:spanner-adaptation-to-small-memory}
	Recall that to sample the values $\delta_{v}$, each node $v \in V$ has to store $O (\kappa)$ bits in its memory. We now present a simple modification to the sampling procedure that accommodates a memory space of $O (\log \kappa)$ for every node.
	As a consequence, we get an algorithm where each node $v \in V$ uses $O (\deg(v) + \log \kappa)$ memory space.
	The idea is for each node $v \in V$ to run the sampling procedure for $\kappa - 1$ iterations denoted by $0, \dots, \kappa - 2$, where in the $i$-th iteration, $v$ executes only the $i$-th experiment (i.e., tosses only the coins associated with the $i$-th experiments of each of the basic scheme executions and determines its success/failure accordingly).
	Then, $v$ selects $\delta_{v}$ to be the index of the first successful experiment if one exists, or $\kappa - 1$ otherwise. We note that one can easily adapt the correctness arguments presented for \Thm{}~\ref{thm:spanner} to this modified version.
	The runtime becomes $O (\kappa \log n)$ whp due to the runtime of the sampling procedure.
	Hence, the following theorem is obtained. 
	
	\begin{theorem}\label{thm:spanner-low-memory}
		There exists an algorithm in the GRC model that computes a set $H\subseteq E$ of edges such that $H$ is a $(2 \kappa - 1)$-spanner whp, and
		$\mathbb{E}[|H|] = O (n^{1 + \frac{1 + \varepsilon}{\kappa}})$.
		The runtime of the algorithm is $O (\kappa \log n)$ rounds whp, and the memory space used by each node $v \in V$ is $O (\deg (v) + \log \kappa)$.
	\end{theorem}

	\section{Verification Tasks} 
	\label{sec:verification-tasks}
	In this section, we provide GRC algorithms for various \emph{verification tasks}. We note that all of these tasks were previously studied by Das Sarma et al.~\cite{dasHKKNPPW2011} in the context of lower bounds in the CONGEST model. In verification tasks, the goal is to decide whether a connected graph $G=(V,E)$ and an input assignment $\mathcal{I}:V\rightarrow\{0,1\}^{*}$ satisfy a certain property. Formally, we represent verification tasks as a predicate $\Pi$ such that $\Pi(G,\mathcal{I})=1$ if the property in question is satisfied by $G$ and $\mathcal{I}$; and $\Pi(G,\mathcal{I})=0$ otherwise. For all of the verification tasks in this section, the input assignment $\mathcal{I}$ encodes a subgraph $H = (V_{H}, E_{H})$ in a distributed manner (possibly among other input components relevant to the task at hand). The correctness requirement for an algorithm that decides a predicate $\Pi$ is that all nodes output a correct answer.\footnote{Notice that this is a stronger requirement than the standard where for `no' instances, it suffices that some of the nodes output a negative answer. However, in the GRC model, this stronger correctness requirement can be obtained at the cost of at most one additional round. This is because any node that outputs a negative answer can inform all other nodes via a global circuit.} 
	
	\subsection{Minimum Spanning Tree Verification}
	\label{sec:mst-verification}
	The MST predicate $\Pi$ is defined as follows.
	For every graph $G = (V, E)$, an edge-weight function
	$w : E \rightarrow \{ 1, \dots, W \}$
	for some integer $W>0$,
	and a subgraph $H = (V_{H}, E_{H})$ of $G$, the predicate satisfies $\Pi(G, w, H) = 1$ if and only if $H$ is an MST of $G$ with respect to $w$.
	
	Recall that
	$W = \max_{e \in E} \{ w(e) \}$.
	The MST verification algorithm $\Alg$ runs in
	$O (\log n \cdot \log (n + W))$
	rounds whp and works as follows.
	The nodes compute a new edge-weight function $w': E \rightarrow \{1,\dots,2W\}$ defined as
	\[
	\begin{cases}
		w'(e) = 2w(e)-1 ,
		&
		e \in E_{H}
		\\
		w'(e) = 2w(e)
		\, ,
		&
		\text{otherwise}\ 
	\end{cases}
	\, .
	\]
	Notice that $w'$ satisfies
	\[
	w(e_{1}) > w(e_{2}) \Rightarrow w'(e_{1}) > w'(e_{2})
	\]
	for every two edges $e_{1}, e_{2} \in E$.
	Next, the nodes run the MST algorithm described in \Sect{}~\ref{sec:mst} on $G$ and $w'$.
	Every node $v \in V$ then checks if $T(v) = E_{H}(v)$ where $T$ is the output of the MST algorithm.
	If not, $v$ beeps through the global circuit, and all nodes output a negative answer.
	If the global circuit is silent during this last round, all nodes output a positive answer.
	
	From the definition of $w'$, we obtain the following observation.
	\begin{observation}
		\label{obs:mst-verification}
		For every graph $G$ and an edge-weight function $w$, if $T$ is an MST of $G$ with respect to $w'$, then $T$ is an MST of $G$ with respect to $w$.
	\end{observation}
	We now prove the correctness of $\Alg$.
	\begin{lemma}
		\label{lem:mst-verification}
		Given a graph $G = (V, E)$, an edge-weight function
		$w : E \rightarrow \Reals$,
		and a subgraph $H = (V_{H}, E_{H})$ of $G$,
		the MST verification algorithm verifies that $H$ is an MST of $G$ whp.
	\end{lemma}
	\begin{proof}
		From \Obs{}~\ref{obs:mst-verification} combined with \Thm{}~\ref{thm:mst}, it follows that the nodes output a positive answer only if $H$ is an MST of $G$.
		In the converse direction, assume by contradiction that the nodes output a negative answer while $H$ is an MST of $G$.
		This means that the MST algorithm outputs a tree $T \neq E_{H}$.
		From the definition of $w'$, it follows that $w(T) < w(E_{H})$, in contradiction to the assumption.
	\end{proof}
	
	Recalling that the run time of the MST algorithm is
	$O (\log n \cdot \log (n + W))$
	(see \Lem{}~\ref{lem:mst-runtime}), we establish the following theorem.
	\begin{theorem}
		\label{thm:mst-verification}
		There exists an algorithm in the GRC model whose runtime is $O (\log n \cdot \log (n + W))$ rounds whp that verifies the MST predicate whp.
	\end{theorem}
	
	\subsection{Additional Verification Tasks}
	\label{sec:additional-verification-tasks}
	\subparagraph{Connected Spanning Subgraph.}
	The \emph{connected spanning subgraph} predicate $\Pi$ is defined as follows.
	For every graph $G = (V, E)$ and a subgraph $H = (V_{H}, E_{H})$ of $G$, the predicate satisfies $\Pi(G, H) = 1$ if and only if
	(1)
	$V_{H} = V$; and
	(2)
	$H$ is connected.
	
	The connected spanning subgraph verification algorithm $\Alg$ runs in $\Theta (\log n)$ rounds whp and works as follows.
	First, the nodes construct a global circuit.
	Then, each node $v \in V$ checks if at least one of its incident edges is in $E_{H}$.
	If not, $v$ beeps through the global circuit, and all nodes output a negative answer.
	Otherwise, the nodes execute the outgoing edge detection procedure described in \Sect{}~\ref{sec:auxiliary} on $G$ and $H$.
	Upon termination of the procedure, if an outgoing edge was detected by some node $v \in V$, then $v$ beeps through the global circuit, and all nodes output a negative answer.
	If the global circuit is silent during this last round, all nodes output a positive answer.
	
	The correctness and runtime of this algorithm follow directly from the design of $\Alg$ and from \Lem{}~\ref{lem:outgoing}-\ref{lem:outgoing-runtime}.
	
	\subparagraph{$e$-Cycle Containment.}
	The \emph{$e$-cycle containment} predicate $\Pi$ is defined as follows.
	For every graph $G = (V, E)$, a subgraph $H = (V_{H}, E_{H})$ of $G$, and an edge $e \in E$, the predicate satisfies $\Pi(G, H, e) = 1$ if and only if
	$H$ contains a cycle containing $e$.
	\begin{observation}
		\label{obs:e-cycle-containment}
		If edge $e \in E$ is an outgoing edge for some cluster induced by $E_{H} - \{ e \}$ on $G$, then $e$ is not contained in a cycle of $H$.
	\end{observation}
	The $e$-cycle containment verification algorithm $\Alg$ runs in $\Theta (\log n)$ rounds whp and works as follows.
	First, the nodes construct a global circuit.
	In the first round, both endpoints of $e$ check if $e \in E_{H}$.
	If not, they beep through the global circuit, and all nodes output a negative answer.
	Otherwise, the nodes execute the outgoing edge detection procedure described in \Sect{}~\ref{sec:auxiliary} on $G$ and $E_{H} - \{ e \}$.
	Upon termination of the procedure, if $e$ is determined as an outgoing edge, then both its endpoints beep through the global circuit, and all nodes output a negative answer. 
	If the global circuit is silent during this last round, all nodes output a positive answer.
	The correctness and runtime of this algorithm follow directly from
	\Obs{}~\ref{obs:e-cycle-containment} and from
	\Lem{}~\ref{lem:outgoing}-\ref{lem:outgoing-runtime}.
	
	\subparagraph{$(s, t)$-connectivity.}
	The \emph{$(s, t)$-connectivity} predicate $\Pi$ is defined as follows.
	For every graph $G = (V, E)$, a subgraph $H = (V_{H}, E_{H})$ of $G$, and two nodes $s, t \in V$, the predicate satisfies $\Pi(G, H, s, t) = 1$ if and only if $s$ and $t$ are in the same connected component of $H$.
	
	The $(s, t)$-connectivity verification algorithm $\Alg$ runs in $\Theta (\log n)$ rounds whp and works as follows.
	First, the nodes construct two global circuits and a circuit for each connected component of $H$.
	Through the first global circuit, the nodes execute \Proc{}~$\mathtt{CountingToLogn}$ described in \Sect{}~\ref{sec:auxiliary} for (a sufficiently large) $c > 1$ times.
	While \Proc{}~$\mathtt{CountingToLogn}$ is being executed, nodes $s$ and $t$ toss random bits and beep them through the circuit of their connected component.
	If in some round one of $s$ and $t$ is silent and hears a beep through its component's circuit, it beeps through the second global circuit, and all nodes output a positive answer.
	Upon termination of \Proc{}~$\mathtt{CountingToLogn}$, if the second global circuit was silent throughout the execution, all nodes output a negative answer.
	
	The correctness and runtime of this algorithm follow directly from the design of $\Alg$ and from \Lem{}~\ref{lem:counting}.
	
	\subparagraph{Connectivity.}
	The \emph{connectivity} predicate $\Pi$ is defined as follows.
	For every graph $G = (V, E)$ and a subgraph $H = (V_{H}, E_{H})$ of $G$, the predicate satisfies $\Pi(G, H) = 1$ if and only if $H$ is connected.
	
	The connectivity verification algorithm $\Alg$ runs in $\Theta (\log n)$ rounds whp and works as follows.
	First, the nodes construct a global circuit.
	Then, the nodes execute the outgoing edge detection procedure described in \Sect{}~\ref{sec:auxiliary} on $G$ and $H$.
	Note that a node $v \in V - V_{H}$ does not participate in this stage.
	Upon termination of the procedure, if an outgoing edge was detected by some node $v \in V_{H}$, then $v$ beeps through the global circuit, and all nodes output a negative answer.
	If the global circuit is silent during this last round, all nodes output a positive answer.
	
	The correctness and runtime of this algorithm follow directly from the design of $\Alg$ and from \Lem{}~\ref{lem:outgoing}-\ref{lem:outgoing-runtime}.
	
	\subparagraph{Cut Verification.}
	The \emph{cut verification} predicate $\Pi$ is defined as follows.
	For every graph $G = (V, E)$ and a subgraph $H = (V_{H}, E_{H})$ of $G$, the predicate satisfies $\Pi(G, H) = 1$ if and only if the graph $G' = (V, E - E_{H})$ is not connected.
	
	The cut verification algorithm runs in $\Theta (\log n)$ rounds whp and works as follows.
	We solve this task by a reduction from the connectivity verification task described above.
	The nodes construct the graph
	$G' = (V, E - E_{H})$
	and run the connectivity verification algorithm.
	If the answer is positive, the nodes output a negative answer, and vice versa.
	
	The correctness and runtime of this algorithm follow directly from the correctness and runtime of the connectivity verification algorithm.

	\subparagraph{Edge on all $(s, t)$-Paths.}
	The \emph{edge on all $(s, t)$-paths} predicate $\Pi$ is defined as follows.
	For every graph $G = (V, E)$, a subgraph $H = (V_{H}, E_{H})$ of $G$, and an edge $e = (s, t) \in E_{H}$, the predicate satisfies $\Pi(G, H, e) = 1$ if and only if
	$e$ is an $(s, t)$-cut in $H$, namely, the edge $e$ lies on every path between
	$s$ and $t$ in $H$.
	
	\begin{observation}
		\label{obs:edge-on-all-paths}
		For two nodes $s, t \in V_{H}$, the edge $e = (s, t)$ lies on every path
		between $s$ and $t$ in $H$ if and only if $e$ is not contained in a cycle of
		$H$.
	\end{observation}
	
	The edge on all paths verification algorithm runs in $\Theta (\log n)$ rounds whp and works as follows.
	The nodes run the $e$-cycle containment verification algorithm described above on $G$, $H$, and $e$.
	If the answer is positive, the nodes output a negative answer, and vice versa.
	
	The correctness and runtime of this algorithm follow directly from the correctness of the $e$-cycle containment verification algorithm, combined with \Obs{}~\ref{obs:edge-on-all-paths}.
	
	\subparagraph{$(s, t)$-Cut.}
	The \emph{$(s, t)$-cut} predicate $\Pi$ is defined as follows.
	For every graph $G = (V, E)$, a subgraph $H = (V_{H}, E_{H})$ of $G$, and two nodes $s, t \in V$, the predicate satisfies $\Pi(G, H, s, t) = 1$ if and only if $E_{H}$ is an $(s, t)$-cut in $G$.
	
	The $(s, t)$-cut verification algorithm runs in $\Theta (\log n)$ rounds whp and works as follows.
	We solve this task by a reduction from the $(s, t)$-connectivity verification task described above.
	The nodes construct the graph
	$G' = (V, E - E_{H})$
	and run the $(s, t)$-connectivity verification algorithm on $G$, $G'$, $s$, and $t$.
	If the answer is positive, the nodes output a negative answer, and vice versa.
	
	The correctness and runtime of this algorithm follow directly from the correctness and runtime of the $(s, t)$-connectivity verification algorithm.
	
	\subparagraph{Hamiltonian Cycle.}
	The \emph{Hamiltonian cycle} predicate $\Pi$ is defined as follows.
	For every graph $G = (V, E)$ and a subgraph $H = (V_{H}, E_{H})$ of $G$, the predicate satisfies $\Pi(G, H) = 1$ if and only if
	(1)
	$H$ is a simple cycle; and
	(2)
	$V_{H} = V$.
	
	The Hamiltonian cycle verification algorithm $\Alg$ runs in $\Theta (\log n)$ rounds whp and works as follows.
	First, the nodes construct a global circuit.
	Then, each node $v \in V$ checks if $|E_{H}(v)| = 2$.
	If not, $v$ beeps through the global circuit, and all nodes output a negative answer.
	Otherwise, the nodes execute the outgoing edge detection procedure described in \Sect{}~\ref{sec:auxiliary} on $G$ and $H$.
	Upon termination of the procedure, if an outgoing edge was detected by some node $v \in V$, then $v$ beeps through the global circuit, and all nodes output a negative answer.
	If the global circuit is silent during this last round, all nodes output a positive answer.
	
	The correctness and runtime of $\Alg$ follow directly from the design of $\Alg$ and from \Lem{}~\ref{lem:outgoing}-\ref{lem:outgoing-runtime}.
	
	\subparagraph{Simple Path.}
	The \emph{simple path} predicate $\Pi$ is defined as follows.
	For every graph $G = (V, E)$ and a subgraph $H = (V_{H}, E_{H})$ of $G$, the predicate satisfies $\Pi(G, H) = 1$ if and only if
	$H$ is a simple path.
	\begin{observation}
		\label{obs:simple-path}
		A graph $F = (V_{F}, E_{F})$ is a simple path if all of the following three conditions are satisfied:
		(1)
		for every node $v \in V$ is holds $|E_{F}(v)| \leq 2$;
		(2)
		there exists at least one node $v \in V$ such that $|E_{F}(v)| = 1$; and
		(3)
		$F$ is connected
	\end{observation}
	
	The simple path verification algorithm $\Alg$ runs in $\Theta (\log n)$ rounds whp and works as follows.
	First, the nodes construct a global circuit.
	Then, each node $v \in V$ checks if $|E_{H}(v)| \leq 2$.
	If not, $v$ beeps through the global circuit, and all nodes output a negative answer.
	Next, a node whose degree is $1$ beeps through the global circuit.
	If the global circuit is silent during this round, all nodes output a negative answer.
	Otherwise, the nodes execute the connectivity verification algorithm described above.
	If the answer is positive, all nodes output a positive answer, and vice versa.
	
	The correctness and runtime of $\Alg$ follow directly from the design of $\Alg$ combined with \Obs{}~\ref{obs:simple-path}.
	
	{\vskip 1pc}
	We therefore establish the following theorem.
	\begin{theorem}
		\label{thm:verification-tasks}
		There exist algorithms in the GRC model whose runtime is $\Theta (\log n)$ rounds whp that verify the following predicates whp: connected spanning subgraph, $e$-cycle containment, $(s, t)$-connectivity, connectivity, cut verification, edge on all paths, $(s, t)$-cut, Hamiltonian cycle, and simple path.
	\end{theorem}
	
	\section{Lower Bounds} 
	\label{sec:lower-bounds}
	In this section, we present a generic lower bound for the GRC model. The lower bound relies on a reduction from functions in the (two-party) \emph{communication complexity} setting \cite{Kushilevitz-Nisan}. In the communication complexity setting, two players, namely Alice and Bob, each receive an input string $x\in \{0,1\}^{\lambda}$ and $y\in \{0,1\}^{\lambda}$, respectively. Their goal is to jointly compute some function of $x$ and $y$ by exchanging messages. A notoriously hard (and thus, useful in the context of hardness results) function in communication complexity is \emph{set-disjointness}, which is denoted by $\mathtt{DISJ}(x,y)$ and defined so that $\mathtt{DISJ}(x,y)=1$ if $\langle x, y\rangle=0$ (where $\langle \cdot, \cdot\rangle$ denotes the inner-product of two vectors); and $\mathtt{DISJ}(x,y)=0$ otherwise. A well-known result is that solving set-disjointness requires $\Omega(\lambda)$ bits of communication between Alice and Bob even if the parties have access to a shared random bit-string of unbounded size \cite{Kushilevitz-Nisan}. For ease of presentation, the generic lower bound shown in this section is given based on reductions from set-disjointness. Extending the lower bound to other communication complexity functions can be done in a straightforward manner.
	
	Towards presenting the lower bound, we define the following notion for graph
	decision problems.
	Let $f,h:\mathbb{Z}_{> 0}\rightarrow \mathbb{Z}_{>0}$ be a pair of functions. A graph decision problem $\Pi$ is said to be \emph{$(f,h)$-hard} if for any $\lambda_{0}>0$, there exists an integer $\lambda>\lambda_{0}$ such that for any pair of bit-strings $x,y\in\{0,1\}^{\lambda}$, there exist an $n$-node graph $G(x,y)=(V,E)$, a partition of $V$ into two disjoint non-empty sets $A,B\subset V$, and a bit $b\in\{0,1\}$ such that: (1) the edges of $E$ with both endpoints in $A$ (resp., $B$) are fully determined by $x$ (resp., $y$); (2) the edges that cross the $(A,B)$-cut do not depend on $x$ or $y$; (3) $\lambda\geq f(n)$; (4) $h(n)\geq |\partial_{A}|$;
	and (5) $\Pi(G(x,y))=1\iff\mathtt{DISJ}(x,y)=b$.\footnote{The definition of $(f,h)$-hardness can be naturally extended to graph decision problems that include an input assignment to the nodes.}
	\begin{theorem}\label{thm:lower-bound-reduction}
		If a graph decision problem $\Pi$ is $(f,h)$-hard for functions $f,h:\mathbb{Z}_{> 0}\rightarrow \mathbb{R}_{>0}$, then the runtime of any (randomized) algorithm for $\Pi$ in the GRC model is bounded by $\Omega\left(\frac{f(n)}{h(n)\cdot k\cdot (\log h(n)+\log k)}\right)$, where $k$ is the number of pins assigned to every edge $e\in E$.
	\end{theorem}
	
	\begin{table}
		\begin{center}
			\resizebox{\columnwidth}{!}{%
				\begin{tabular}{| l| l| l| l| l|}
					\hline
					task & \textbf{$f(n)$} & \textbf{$h(n)$} & runtime & paper \\
					\hline
					$W$-weighted $8$-cycle freeness & $\Omega(n^{2})$ & $O(1)$ & $\Omega(n^{2})$ & \cite{AbboudCKP21}\\
					\hline
					$3$-colorability & &&&\\
					minimum vertex cover & $\Omega(n^{2})$ & $O(\log n)$ & $\Omega\left(\frac{n^{2}}{\log n\log \log n}\right)$ & \cite{AbboudCKP21}\\
					maximum independent set &&&&\\
					\hline
					minimum dominating set & $\Omega(n^{2})$ & $O(\log n)$ & $\Omega\left(\frac{n^{2}}{\log n\log \log n}\right)$ & \cite{BachrachCDELP19}\\
					\hline
					diameter $>2$ & $\Omega(n^{2})$ & $O(n)$ & $\Omega(n/\log n)$ & \cite{frischknechtHW2012,Censor-HillelPP20}\\
					\hline
					$C_{5}$-freeness& $\Omega(n^{2})$& $O(n)$ & $\Omega(n/\log n)$ & \cite{DruckerKO13}\\
					\hline
					radius $>3$ & $\Omega(n)$ & $O(\log n)$ & $\Omega\left(\frac{n}{\log n\log \log n}\right)$ & \cite{AbboudCKP21}\\
					\hline
					$C_{4}$-freeness & $\Omega(n^{3/2})$& $O(n)$ & $\Omega(\sqrt{n}/\log n)$ & \cite{DruckerKO13}\\
					\hline
					$K_{4}$-freeness& $\Omega(n^{2})$& $O(n^{3/2})$ & $\Omega(\sqrt{n}/\log n)$ & \cite{CzumajK20}\\
					\hline
				\end{tabular}%
			}
		\end{center}
		\caption{\label{table:lower-bounds}%
			A sample of existing
			$(f, h)$-hardness
			results and the GRC runtime lower bounds derived from them (assuming that each
			edge is associated with
			$k = O (1)$
			pins).}
	\end{table}
	
	Before proving \Thm{}~\ref{thm:lower-bound-reduction}, let us present its
	applicability.
	Reductions from communication complexity functions to graph problems have been
	explored thoroughly in the context of lower bounds for CONGEST algorithms (see, e.g., \cite{censorD2018,frischknechtHW2012,BachrachCDELP19,AbboudCKP21,GrossmanKP20,Censor-HillelKP17,AbboudCK16,DruckerKO13,CzumajK20}).
	Consequently, many natural graph problems admit non-trivial $(f,h)$-hardness
	results, and thus,  \Thm{}~\ref{thm:lower-bound-reduction} establishes a lower bound for these problems in the context of the GRC model.
	Refer to Table \ref{table:lower-bounds} for a sample of concrete GRC runtime
	lower bounds that follow
	from \Thm{}~\ref{thm:lower-bound-reduction}.
	
	We go on to prove \Thm{}~\ref{thm:lower-bound-reduction}.
	\begin{proof}[Proof of \Thm{}~\ref{thm:lower-bound-reduction}]
		Let \Alg{}~be an algorithm for $\Pi$ in the GRC model. Given inputs $x,y\in \{0,1\}^{\lambda}$ and a shared random bit-string, we show that Alice and Bob can simulate \Alg{}~on the graph $G(x,y)=(V=A\dot{\cup}B,E)$ to decide $\Pi(G(x,y))$ and thus solve set-disjointness. As standard, the shared random bit-string is utilized to simulate the random coins tossed by the nodes during \Alg{}. Let $Q=\partial_{A}\times [k]$ be the set of pins associated with edges crossing the $(A,B)$-cut. We henceforth assume that the pins of $Q$ are ordered in a manner agreed upon by Alice and Bob. In our simulation, Alice and Bob exchange
		$O(|Q|\log |Q|)$
		bits to simulate a single round. Since set-disjointness requires $\Omega(\lambda)$ bits of communication, this implies that \Alg{}~terminates after $\Omega\left(\frac{\lambda}{|Q|\log |Q|}\right) $ rounds. Because $f(n)\leq \lambda$ and $h(n)\cdot k\geq |\partial_{A}|\cdot k=|Q|$, this leads to the desired bound of $\Omega\left(\frac{f(n)}{h(n)\cdot k\cdot (\log h(n)+\log k)}\right)$ rounds.   
		
		We describe the simulation of round $t$. Alice simulates the nodes of $A$, and Bob simulates the nodes of $B$. The simulation is presented from Alice's side as Bob's side is analogous. Let $\Pins_{A}=\bigcup_{v\in A}\Pins(v)$ be the pins incident on the nodes of $A$. Suppose that Alice knows the states of all nodes in $A$ (including their pin partition) at the beginning of round $t$. The goal of the simulation is for Alice to know for each pin $p\in P_{A}$ if a beep occurred on its circuit in round $t$ (which would allow Alice to compute the states of all nodes in $A$ in round $t+1$). 
		
		Define $\mathcal{L}_{A}^{t}$ to be the projection of the relation $\mathcal{L}^{t}$ over the nodes of $A$. That is, $\mathcal{L}_{A}^{t}$ is the symmetric binary relation over $\Pins_{A}$ such that pins $p = (e, i)$ and $p' = (e', i')$ belong to $\mathcal{L}_{A}^{t}$ if and only if there exists a node $v \in A$ (incident on both $e$ and $e'$) such that $p$ and $p'$ belong
		to the same set $R\subseteq \mathcal{R}^{t}(v)$. Let $\RefTrnsCl(\mathcal{L}^{t}_{A})$ be the reflexive transitive closure of $\mathcal{L}_{A}^{t}$. For a pin $p\in P_{A}$, let $\Circuits_{A}^{t}(p)$ denote the equivalence class of $\RefTrnsCl(\mathcal{L}^{t}_{A})$ to which $p$ belongs. Observe that Alice can compute the relations $\mathcal{L}_{A}^{t}$ and $\RefTrnsCl(\mathcal{L}^{t}_{A})$ by herself. Furthermore, if an equivalence class of $\RefTrnsCl(\mathcal{L}^{t}_{A})$ contains only pins from $P_{A}-Q$, then it forms a circuit $C$ in round $t$ where all the nodes partaking in $C$ are from $A$. Therefore, in this case, Alice can compute whether a beep was transmitted on the circuit $C$ by herself. In the case of circuits containing the pins of $Q$, communication with Bob is needed (since nodes from $B$ partake in these circuits).  
		
		
		The communication scheme is defined as follows. Alice starts by giving a unique $O(\log |Q|)$-bit \emph{name} to each equivalence class in $\{\Circuits_{A}^{t}(p)\mid p\in Q\}$. Then, Alice sends Bob a message (and vice versa) where for each pin $p_{i}$ (indexed according to the predetermined order), she writes the name of $\Circuits_{A}^{t}(p_{i})$ and a bit indicating if a beep was sent through $p_{i}$ from Alice's side (i.e., on any of the pins in $\Circuits_{A}^{t}(p_{i})$). Notice that by definition, two pins $p,p'\in Q$ belong to the same circuit in round $t$ if and only if there exists a sequence $p_{1}=p,p_{2},\dots,p_{\ell}=p'\in Q$ of pins such that $p_{i}$ and $p_{i+1}$ were given the same name in either Alice or Bob's message for all $i\in [\ell-1]$. Therefore, given Bob's message, Alice can partition the pins of $Q$ according to their circuits in round $t$ and determine whether a beep was transmitted on these circuits. Furthermore, since Alice knows the local pin partition of all nodes in $A$, she can partition all the pins in $P_{A}$ according to their circuits in round $t$ and determine whether a beep was transmitted on these circuits. This means that Alice successfully simulates round $t$. The number of bits exchanged between Alice and Bob to obtain the simulation of a round is $O(|Q|\log |Q|)$. As discussed above, this directly implies an $\Omega\left(\frac{f(n)}{h(n)\cdot k\cdot (\log h(n)+\log k)}\right)$ lower bound on the runtime of \Alg{}.  
	\end{proof}
	
	\subsection{Inapplicable Reductions}\label{sec:inapplicable-reductions}
	As mentioned before, reductions from communication complexity functions to graph problems have been widely studied, particularly in the context of lower bounds for CONGEST algorithms. While most of the reductions in the literature give meaningful $(f,h)$-hardness results according to the notion presented above, some do not. An example of reductions that do not yield meaningful $(f,h)$-hardness results are the reductions presented in \cite{dasHKKNPPW2011}. These reductions are used to show lower bounds for the runtime of CONGEST algorithms for various problems, including minimum spanning tree, connected spanning subgraph verification, and cut verification. However, as we will discuss soon, these lower bounds do not apply to the GRC model. Moreover, in \Sect{}~\ref{sec:mst} and \Sect{}~\ref{sec:verification-tasks} we show runtime upper bounds in the GRC model of $O(\log ^{2}n)$ and sometimes even $O(\log n)$ for many of the problems discussed in \cite{dasHKKNPPW2011}.

	We now go over the construction of \cite{dasHKKNPPW2011} and explain why it does not apply to algorithms in the GRC model. Let us use the connected spanning subgraph verification as an example (see \Sect{}~\ref{sec:additional-verification-tasks} for a definition). Given inputs $x,y\in \{0,1\}^{\lambda}$, Alice and Bob construct graph $G(x,y)$ with subgraph $H(x,y)$ such that $H(x,y)$ is a connected spanning subgraph of $G(x,y)$ if and only if $\mathtt{DISJ}(x,y)=1$. As part of the graph construction, $G(x,y)$ contains $\lambda$ simple paths $\mathcal{P}_{1},\dots ,\mathcal{P}_{\lambda}$, each of length $\ell$, where initially, both Alice and Bob know the states of all nodes on every path (refer to \cite{dasHKKNPPW2011} for the full construction). For convenience, suppose that each path forms a horizontal line from the leftmost to the rightmost node. An invariant maintained throughout the simulation is that for every path $\mathcal{P}_{i}$, following round $t$'s simulation, Alice (resp., Bob) knows the states of all $\ell-t$ leftmost (resp., rightmost) nodes at the end of round $t$ (in contrast to the simulation that appears in \Thm{}~\ref{thm:lower-bound-reduction}, where each party simulates a static set of nodes). An important observation that enables succinct messages between Alice and Bob is that when simulating round $t+1$, Alice can compute the message from the $\ell-t$ leftmost node to the $\ell-t-1$ leftmost node on each path by herself (since she knows the state of the $\ell-t$ leftmost node at the beginning of the round). Similarly, Bob can compute the message from the $\ell-t$ rightmost node to the $\ell-t-1$ rightmost node on each path by himself. That is, throughout the simulation, the communication does not need to account for the messages sent along the paths $\mathcal{P}_{1},\dots ,\mathcal{P}_{\lambda}$ during the simulated CONGEST algorithm. 
	
	We note that the simulation described cannot truthfully depict an algorithm in the GRC model. Generally speaking, this is because the GRC model allows for non-neighbors to communicate. For example, suppose that at the beginning of round $t\geq 2$ of the GRC algorithm, the nodes of some path $\mathcal{P}_{i}$ all partake in some circuit $C$. In the simulation, Alice does not know the state of the rightmost node (as it is too far away from the $\ell-t\leq \ell-2$ leftmost node) and thus cannot compute by herself whether the rightmost node sent a beep through the circuit $C$ in round $t$. Extending this observation, any attempt to simulate an algorithm in the GRC model on the graph $G(x,y)$ requires at least $\lambda$ bits of communication per round (one bit per path). Since $\lambda$ bits are sufficient to solve set-disjointness, a non-trivial lower bound cannot be derived using the construction of \cite{dasHKKNPPW2011}.
	
	\section{Additional Related Work}
	\label{sec:related-work}
	The geometric amoebot model provides an abstraction for distributed computing
	by (finitely many) computationally restricted particles that can move in the
	hexagonal grid by means of expansions and contractions.
	The model was introduced by Derakhshandeh et al.~\cite{derakhshandehDGRSS2014}
	and gained considerable popularity since then, see, e.g.,
	\cite{derakhshandehGSBRS2015,
		derakhshandehGRSS2015,
		cannonDRR2016,
		daymudeGRSS2017,
		DerakhshandehGRST2017universal-coating,
		DaymudeHRS2019programmable,
		DaymudeRS2023canonical}.
	The notion of reconfigurable circuits, from which our GRC model is derived,
	was introduced by Feldmann et al.~\cite{feldmannPSD2022}, and studied further
	by Padalkin et al.\
	\cite{padalkinSW2022,
		PadalkinS2024shortest-path},
	as an augmentation of the geometric amoebot model with the capability to form
	long-range (reconfigurable) beeping channels.
	Since geometric amoebot algorithms are confined to the hexagonal grid, so are
	the algorithms presented in
	\cite{feldmannPSD2022,
		padalkinSW2022,
		PadalkinS2024shortest-path},
	only that unlike the former algorithms, that often exploit the particles'
	mobility, the latter algorithms are static.
	As such, the reconfigurable circuits algorithms of 
	\cite{feldmannPSD2022,
		padalkinSW2022,
		PadalkinS2024shortest-path}
	operate under a special case of our GRC model, restricted to (finite subgraphs
	of) the hexagonal grid.\footnote{%
		The geometric amoebot model supports a different type of communication scheme
		for short-range interactions, where a particle observes the full state of each
		of its neighboring particles.
		We did not include this type of communication scheme in the GRC model that
		uses the same mechanism for long-range and short-range communication.}
	In contrast to the distributed tasks studied in the current paper which are
	``purely combinatorial'', most of the tasks studied in
	\cite{feldmannPSD2022,
		padalkinSW2022,
		PadalkinS2024shortest-path}
	admit a ``geometric flavor'' corresponding to an underlying planar embedding
	of the hexagonal grid;
	these tasks include
	compass alignment,
	chirality agreement,
	shape recognition,
	stripe computation,
	and
	identifying the northern-most node.
	Non-geometric exceptions are leader election and the construction of shortest
	paths, however the algorithms developed in
	\cite{feldmannPSD2022,
		padalkinSW2022,
		PadalkinS2024shortest-path}
	for these tasks are tailored to the hexagonal grid and strongly rely on its
	unique features.
	
	Another aspect in which the general GRC model deviates from the ``special
	case'' considered in \cite{feldmannPSD2022} is the exact meaning of
	uniformity:
	Since port numbering is inherent to the communication scheme of the GRC model,
	the state machine associated with a node $v$ must be adjusted to the degree
	$\Degree(v)$ of $v$.
	The degrees in subgraphs of the hexagonal grid are at most $6$, which means
	that the description of algorithms operating under the model of
	\cite{feldmannPSD2022} (or the state machines thereof) can be bounded by a
	universal constant.
	In contrast, the degrees in general graphs may obviously grow asymptotically,
	hence we cannot hope to bound the descriptions of our state machines in the
	same way.
	Nevertheless, these descriptions are still bounded independently of
	any global parameter of the graph $G$.
	
	A computational model for distributed graph algorithms that supports arbitrary
	graph topologies and does admit (universally) fixed state machines is (what
	came to be known as) the \emph{stone age} model that was introduced by Emek
	and Wattenhofer in \cite{EmekW2013stone-age} and studied further, e.g., in
	\cite{AfekEK2018leader-selection,
		AfekEK2018synergy,
		EmekU2020dynamic-sa,
		EmekK2021au,
		GiakkoupisZ2023two-state-process}.
	In this model, however, the nodes have no direct access to their incident
	edges, hence stone age algorithms are not suitable for \emph{edge sensitive
		tasks}, namely, tasks whose input and/or output may distinguish between the
	graph's edges (such as the tasks addressed in the current paper).
	
	The beeping communication scheme for distributed graph algorithms was
	introduced by Cornejo and Kuhn in \cite{CornejoK2010deploying} and studied
	extensively since then, see, e.g.,
	\cite{AfekABCHK2013beeping,
		forsterSW2014,
		gilbertN2015,
		yuJYLC2015,
		czumajD2019,
		Dufoulon2019thesis,
		DufoulonBB2020uncoordinated}.
	These papers assume that each node shares a (static) beeping channel with its
	graph neighbors, in contrast to the GRC model, where the beeping channels are
	reconfigurable and may include nodes from different regions of the graph.
	For the most part, the existing beeping literature does not cover edge
	sensitive tasks (as defined in the previous paragraph).
	The one exception (we are aware of) in this regard is the recent work of
	Dufoulon et al.~\cite{DufoulonEG2023shortest-paths} that designs beeping
	algorithms for the construction of shortest paths, using locally unique node
	identifiers to mark the edges along the constructed paths.
	Notice that the computation/communication of (locally or globally) unique node
	identifiers is inherently impossible when it comes to boundedly uniform
	distributed algorithms, justifying our choice to adopt the port numbering
	convention.
	
	\clearpage
	\bibliographystyle{plainurl}
	\bibliography{references}
	
	\clearpage
	{\centering
		\Large{APPENDIX}
		\par}
	
	\appendix
	
	\section{Missing Proofs}
	\label{appendix:missing-proofs}
	
	\subsection{Proofs Missing from \Sect{}~\ref{sec:preliminaries}}
	\label{appendix:missing-proofs:preliminaries}
	
	\begin{proof}[Proof of \Lem{}~\ref{lem:counting}]
		Consider an execution of $\mathtt{CountingToLogn}$ and integer $t\geq 0$, and let $X_{t}$ be a random variable counting the number of competitors remaining after round $t$, where
		$X_{0} = n$.
		Observe that $\mathbb{E}[X_{t + 1}\mid X_{t}]=X_{t}/2$. By the law of total expectation, it follows that $\mathbb{E}[X_{t + 1}] = \mathbb{E}[\mathbb{E}[X_{t + 1}\mid X_{t}]]= \mathbb{E}[X_{t}/2] = \mathbb{E}[X_{t}]/2$. By a simple inductive argument, we deduce that $\mathbb{E}[X_{t}] = n / 2^t$.
		Thus, for
		$t \geq (1 + \rho ) \log n$,
		it follows that
		$\mathbb{E}[X_{t}] \leq n / 2^{(1 + \rho) \log n} = n^{-\rho}$.
		By Markov's inequality, we get that
		$\Pr[X_{t} \geq 1] \leq n^{-\rho}$
		for every
		$t\geq(1 + \rho) \log n$.
		Since $\tau$ is the median runtime, the probability that
		$\tau\geq(1 + \rho) \log n$
		is equal to the probability that $r$ independent executions of $\mathtt{CountingToLogn}$ took at least $(1 + \rho) \log n$ rounds. Due to the independence of the executions, this probability is at most $n^{-\rho r}$.
		
		We go on to bound the probability that
		$\tau \leq(1 - \rho) \log n$.
		Let $0<\delta<1$ be a constant that satisfies
		$2 / (1 - \delta) < 2^{1 / (1 - \rho)}$
		(notice that such constant exists since $2^{1 / (1 - \rho)}>2$).
		Now, suppose that $X_{t} = x$.
		Notice that $X_{t + 1}$ is a sum of $n$ independent Bernoulli variables with $\mathbb{E}[X_{t + 1}]=x/2$.
		Thus, Chernoff bound implies that
		\[
		\Pr[X_{t + 1} \leq (1 - \delta) \cdot x/2] \leq e^{-\delta^{2} x / 4} \, .
		\]
		Hence, if $X_{t}=x \geq (8 \rho \ln n) / \delta^{2}$, then with probability at most $n^{-2 \rho}$, 
		$X_{t + 1} < (1 - \delta) x/2$.
		Furthermore, there are $O (\log n)$ rounds in which this event can happen (since in each such round the number of competitors is reduced by at least a constant fraction).
		Therefore, by a union bound argument, it holds that with probability greater than $1 - n^{-\rho}$, for every round $t$ such that $X_{t}\geq (8 \rho \ln n )/ \delta^{2}$, a $(1 - \delta)/2$ fraction of the competitors remain in round $t + 1$.
		This implies that with probability greater than $1 - n^{-\rho}$, the procedure runs for over
		$\log_{2 / (1 - \delta)} (n - (8 \rho \ln n) / \delta^{2})$
		rounds.
		Since
		$2 / (1 - \delta) < 2^{1 / (1 - \rho)}$,
		there exists a value $n'>0$ such that
		\[
		\log_{2 / (1 - \delta)} (n - (8 \rho \ln n) / \delta^{2})
		\, \geq \,
		\log _{2^{1 / (1 - \rho)}}n = (1 - \rho) \log n 
		\]
		for all $n\geq n'$. Thus, the probability of an execution of \Proc{}~$\mathtt{CountingToLogn}$ finishing in less than $(1 - \rho) \log n$ rounds is bounded by $n^{-\rho}$.
		Since $\tau$ is the median runtime of $2r - 1$ executions, we get that $\Pr[\tau \leq (1 - \rho) \log n] \leq n^{-\rho r}$.
		Overall, we conclude that
		\begin{align*}
			&\Pr[(1 - \rho) \log n \leq \tau \leq (1 + \rho) \log n]
			\, = \, 
			\\ 
			&1-\Pr[\tau< (1 - \rho) \log n]-\Pr[\tau> (1 + \rho) \log n]\geq 1-2n^{-\rho r} \, .\qedhere
		\end{align*}
	\end{proof}
	
	\begin{proof}[Proof of \Thm~\ref{thm:message-passing-simulation}]
		First, observe that if $\{(e, i)\}\in \PinPart^{t}(u)$ and $\{(e, i)\}\in
		\PinPart^{t}(v)$ for some $e = (u, v) \in E$ and $i \in [k]$ in round $t$,
		then $u$ and $v$ partake in a singleton circuit $C_{e}=\{(e, i)\}$ in round $t$. We start by describing the message-passing simulation assuming that each edge $(u, v) \in E$ is oriented and both endpoints know its orientation.
		If an edge $e = (u, v)$ is oriented from $u$ to $v$, then for any integer
		$t\in\mathbb{Z}_{> 0}$,
		only $u$ can beep through $C_{e}$ in rounds $4t-3$ and $4t-2$, while rounds $4t - 1$ and $4t$ are reserved for $v$.
		Each of the endpoints uses its designated rounds to convey a message (or lack thereof) through the circuit $C_{e}$, where two consecutive beeps stand for the message $1$; two consecutive silences stand for the message $0$; and a silence followed by a beep stand for not sending a message. 
		
		To obtain the edge orientation, the nodes engage in the following symmetry-breaking mechanism executed once in a preprocessing stage. For each node $v \in V$, all its incident edges $e \in E(v)$ are initially unoriented.
		The execution proceeds in phases of two rounds. In the first round of each phase, each node $v \in V$ tosses a fair coin.
		If the coin lands heads, then $v$ beeps through $C_{e}$ for every unoriented edge $e \in E(v)$.
		If $v$ did not beep through a circuit $C_{e}$ in the first round of the phase and heard a beep from the other endpoint, then $v$ orients the edge $e$ away from itself and beeps through $C_{e}$ in the second round of the phase.
		If $v$ beeped through $C_{e}$ in the first round and heard a beep from the other endpoint in the second round, then $v$ orients $e$ towards itself.
		If $v$ still has unoriented incident edges at the end of the phase, then it beeps through a global circuit to inform all nodes. The preprocessing stage ends when all nodes have oriented all their incident edges.
		Observe that by standard Chernoff and union bound arguments, this preprocessing stage terminates after
		$O (\log m) = O (\log n)$ rounds whp.
	\end{proof}
	
	\begin{proof}[Proof of \Lem{}~\ref{lem:outgoing}]
		Consider some edge $e = (u, v) \in E$.
		If $u$ and $v$ belong to the same cluster $S$ (i.e., $e$ is not an outgoing edge) and a single cluster leader in $S$ is selected by the leader election algorithm, then by the construction of the outgoing edge detection procedure, $u$ and $v$ will both classify $e$ as a non-outgoing edge.
		Since the leader election algorithm succeeds whp and there are at most $n$ clusters, by applying union bound over the clusters, it follows that every non-outgoing edge is classified correctly whp.
		
		Now, suppose that $(u, v)$ is an outgoing edge.
		This means that whp, $u$ and $v$ have different cluster leaders $\ell_{u}$ and $\ell_{v}$.
		For any integer $b > 0$, the probability of $\ell_{u}$ and $\ell_{v}$ drawing the same $b$ bit sequence is $2^{-b}$.
		Observe that for a desirably large constant $c'$, it holds that $\ell_{u}$ and $\ell_{v}$ draw $c' \log n$ bits whp simply by repeating the $\mathtt{CountingToLogn}$ procedure $c$ times for a sufficiently large constant $c$.
		Hence, $(u, v)$ is classified as outgoing whp by both $u$ and $v$.
		Since there are less than $n^{2}$ outgoing edges in total, we deduce that all outgoing edges are classified correctly whp by means of a union bound over all outgoing edges.
	\end{proof}
	
	\begin{proof}[Proof of \Lem{}~\ref{lem:outgoing-runtime}]
		Follows directly from \Lem{}~\ref{lem:counting} and \Thm{}~\ref{thm:leader-election}.
	\end{proof}
	
	\subsection{Proofs Missing from \Sect{}~\ref{sec:mst}}
	\label{appendix:missing-proofs:mst}
	
	\begin{proof}[Proof of \Lem{}~\ref{lem:mst-connected-components}]
		If $q_{i} = 1$, then by \Lem{}~\ref{lem:outgoing}, no outgoing edges are detected in phase $i$ whp, and thus the algorithm terminates.
		Otherwise, $(V, T_{i})$ has $q_{i} > 1$ connected components, and by \Lem{}~\ref{lem:outgoing}, each connected component detects all its outgoing edges whp.
		By the construction of the algorithm, this means that each cluster merges with at least one other cluster, and thus the number of clusters is reduced by at least half.
	\end{proof}

	\subsection{Proofs Missing from \Sect{}~\ref{sec:spanner}}
	\label{appendix:missing-proofs:spanner}
	
	\begin{proof}[Proof of \Lem{}~\ref{lem:spanner-capped}]
		Let $I_{M} =\{i \in I \mid X_{i} - q_{i} = M \}$
		and
		$I_{M - 1} =\{i \in I \mid X_{i} - q_{i} = M - 1 \}$.
		We start by showing that
		$\Pr[|I_{M}| \geq t] \leq \phi^{t - 1}$.
		Denote by $Y_{t}$ the $t$-th largest random variable from the values $\{X_{i} - q_{i}\}$.
		Conditioning on $Y_{t} = a$, the event $|I_{M}| \geq t$ is the event that the random variables $Y_{1}, \dots, Y_{t - 1}$ are all equal to $a$ given that they are at least $a$.
		By the memoryless property of the capped geometric distribution between $0$ and $\kappa - 2$, it follows that
		$\Pr[Y_{i} = a\mid Y_{i}\geq a] = \Pr[X_{i} = a + q_{i}\mid X_{i} \geq a + q_{i}] \leq \Pr[X_{i} =0] = \phi$.
		By the independence of the $X_{i}$ values, it follows that $\Pr[|I_{M}| \geq t\mid Y_{t} = a] \leq\phi^{t - 1}$.
		Using the law of total probability, we get
		$\Pr[|I_{M}| \geq t] \leq \phi^{t - 1}$.
		For similar reasoning, we can deduce that $\Pr[|I_{M - 1}| \geq t] \leq\phi^{t - 1}$.
		The statement follows since
		\begin{align*}
			&\mathbb{E}[|I|]
			\, = \,
			\mathbb{E}[|I_{M - 1}|] + \mathbb{E}[|I_{M}|]
			\, = \, \\
			&\sum_{t = 1}^{n}\Pr[|I_{M - 1}| \geq t]+\Pr[|I_{M}| \geq t]
			\, \leq \, 2\cdot\sum_{t=0}^{\infty}\phi^{t}
			\, = \,
			\frac{2}{1 - \phi} \, .\qedhere
		\end{align*}
	\end{proof}

	\begin{proof}[Proof of \Lem{}~\ref{lem:spanner-B}]
		Notice that by design, if the cluster IDs drawn at the bridging edges addition stage are unique, then event $B$ occurs.
		Let $\tau'$ be the random variable counting the number of rounds in the bridging edges addition stage (which is also the length of the cluster IDs).
		First, let us condition on the event $\tau'\geq (c + 2) \log n$. Since there are at most $n$ clusters, the probability of a specific cluster having a non-unique ID in that case is at most $n/(2^{(c + 2) \log n}) = n^{-c - 1}$.
		Applying union bound over all clusters yields
		\begin{align*}
			&\Pr[\lnot B\mid \tau'\geq (c + 2) \log n]
			\, \leq \,\\
			&
			\Pr[\text{IDs are not unique}\mid \tau'\geq (c + 2) \log n]
			\, \leq \, n \cdot n^{-c-1}
			\, = \,
			n^{-c} \, .
		\end{align*}
		
		\Lem{}~\ref{lem:counting} suggests that with probability larger than $1 - 2n^{-c}$, the median runtime of $4c + 7$ calls to $\mathtt{CountingToLogn}$ is at least $\frac{1}{2}\log n$, i.e., there are $2c + 4$ calls whose runtime is at least $\frac{1}{2}\log n$.
		Therefore,
		$\Pr[\tau' \geq (c + 2) \log n] = \Pr[\tau' \geq (2c + 4) \cdot \frac{1}{2} \log n] \geq 1 - 2n^{-c}$.
		
		By the law of total probability, it follows that
		\begin{align*}
			& \Pr[\lnot B]
			\, = \, \Pr[\lnot B \mid \tau'\geq (c + 2) \log n] \cdot \Pr[ \tau' \geq (c + 2) \log n] \, +\, \\
			& \Pr[\lnot B \mid \tau' < (c + 2) \log n]\cdot \Pr[\tau'< (c + 2) \log n] \leq \\
			& \Pr[\lnot B \mid \tau' \geq (c + 2) \log n] + \Pr[\tau'< (c + 2) \log n]\leq \\ 
			& n^{-c} + 2n^{-c} = 3n^{-c}\ ,
		\end{align*}
		which concludes the argument.
	\end{proof}

\end{document}